\documentclass[]{llncs}

\usepackage{amsmath}
\usepackage{amsfonts}
\usepackage{amssymb}
\usepackage{multirow}
\usepackage{graphicx}
\usepackage{lmodern}

\usepackage{epsfig,latexsym}
\usepackage{multicol,tabularx,capt-of}
\usepackage{multirow}
\usepackage{amsmath}

\newtheorem{statement}{Statement}

\usepackage{color}
\newcommand{\commentout}[1]{}

\makeatletter
\newcommand{\ie}{i.\,e.\@ifnextchar{,}{}{~}}
\newcommand{\eg}{e.\,g.\@ifnextchar{,}{}{~}}
\makeatother

\newcommand{\diam}{\mathop{\rm diam} }

\title{Fast approximation of centrality and distances in hyperbolic graphs}
\author{
   Victor Chepoi\inst{1} \and
   Feodor F. Dragan\inst{2} \and
   Michel Habib\inst{3}\and
   Yann Vax\`es\inst{1} \and
   Hend Al-Rasheed\inst{2}
}

\institute{
Laboratoire d'Informatique et Syst\`emes,\\
Aix-Marseille Univ, CNRS, and Univ. de Toulon\\
Facult\'e des Sciences de Luminy,
F-13288 Marseille Cedex 9, France \\
\email{\{victor.chepoi, yann.vaxes\}@lif.univ-mrs.fr}
\and
    Algorithmic Research Laboratory,
    Department of Computer Science, \\
    Kent State University,
    Kent, Ohio, USA  \\
    \email{dragan@cs.kent.edu, halrashe@kent.edu}
\and
Institut de Recherche en Informatique Fondamentale, \\
University Paris Diderot - Paris7,
F-75205 Paris Cedex 13, France \\
\email{habib@liafa.univ-paris-diderot.fr}
}

\begin{document}
\maketitle

\begin{abstract} We show that the eccentricities (and thus the centrality indices) of all vertices of a $\delta$-hyperbolic graph $G=(V,E)$ can be computed in
linear time with an additive one-sided error of at most $c\delta$, i.e.,  after a linear time preprocessing, for every vertex $v$ of $G$ one can compute in $O(1)$ time an estimate $\hat{e}(v)$ of its
eccentricity $ecc_G(v)$ such that $ecc_G(v)\leq \hat{e}(v)\leq ecc_G(v)+ c\delta$ for a small constant $c$. We prove  that every $\delta$-hyperbolic graph $G$
has a shortest path tree, constructible in linear time, such that for every vertex $v$ of $G$, $ecc_G(v)\leq ecc_T(v)\leq ecc_G(v)+ c\delta$. These results are based
on  an interesting monotonicity property of the eccentricity function of hyperbolic graphs: the closer a vertex is to the center of $G$, the smaller its eccentricity is.
We also show that the distance matrix
of $G$ with an additive one-sided error of at most $c'\delta$ can be computed in $O(|V|^2\log^2|V|)$ time, where $c'< c$ is a small constant.
Recent empirical studies show that many real-world graphs (including Internet application
networks, web networks, collaboration networks, social networks, biological networks,
and others) have small hyperbolicity. So, we analyze the performance of our algorithms for approximating centrality  and distance matrix on
a number of real-world networks. Our experimental results show that the obtained estimates are even better than the theoretical bounds.
\end{abstract}


\section{Introduction}

The {\it diameter} $diam(G)$ and the {\it radius} $rad(G)$ of a graph $G=(V,E)$ are two fundamental metric parameters that have many important practical applications
in real world networks. The problem of finding the {\it center} $C(G)$ of a graph $G$ is often studied as a facility location problem for networks where one needs to select a single vertex to
place a facility so that the maximum distance from any demand vertex in the network is minimized. In the analysis of social networks
(e.g., citation networks or recommendation networks), biological systems (e.g., protein interaction networks),
computer networks (e.g., the Internet or peer-to-peer networks), transportation networks (e.g., public transportation
or road networks), etc., the {\it eccentricity} $ecc(v)$ of a vertex $v$ is used to measure the importance of $v$ in the network:
the {\em centrality index} of $v$ \cite{Brandes} is defined as $\frac{1}{ecc(v)}$.

Being able to compute efficiently the diameter, center, radius, and vertex centralities of a  given graph  has become an
increasingly important problem in the analysis of large networks. The algorithmic complexity of the diameter and radius
problems is very well-studied. For some special classes of graphs there are efficient algorithms
\cite{AWV16,++Tamir2007,BrChDr98,Cabello2017,ChDr1994,ChDrVa2002,CoDrHaPa,DrNi2000,DvHa,Hakimi,Olariu}.
However, for general graphs, the only known algorithms computing the diameter and the radius
exactly compute the distance between every pair of vertices in the graph, thus solving the all-pairs shortest
paths problem (APSP) and hence computing all eccentricities.  In view of recent negative results \cite{AWV16,BCH14,RV13},
this seems to be the best what one can do since even for graphs with $m=O(n)$ (where $m$ is the number of edges and $n$ is the number of vertices) the existence of a subquadratic time
(that is, $O(n^{2-\epsilon})$ time for some $\epsilon>0$) algorithm for the diameter or the radius problem will refute
the well known Strong Exponential Time Hypothesis (SETH). Furthermore, recent work \cite{AGV15} shows that if the
radius of a possibly dense graph ($m=O(n^2)$) can be computed in subcubic time ($O(n^{3-\epsilon})$ for some $\epsilon> 0$),
then APSP also admits a subcubic algorithm. Such an algorithm for APSP has long eluded researchers, and it is often conjectured
that it does not exist (see, e.g., \cite{RZ12,V-WW10}).

Motivated by these negative results, researches started devoting more attention to development of fast approximation algorithms. In the analysis of
large-scale networks, for fast estimations of diameter, center, radius, and centrality indices, linear or almost linear time algorithms
are desirable. One hopes also for the all-pairs shortest paths problem to have $o(nm)$ time small-constant--factor approximation algorithms.
In general graphs, both diameter and radius can be 2-approximated by a simple linear time algorithm which picks any node and reports its eccentricity.
A 3/2-approximation algorithm for the diameter and the radius which runs in $\tilde{O}(mn^{2/3})$\footnote{$\tilde{O}$ hides a polylog factor.}
time was recently obtained in \cite{ChLaRoScTaVW2014} (see also \cite{ACIM99} for an earlier $\tilde{O}(n^2 + m\sqrt{n})$ time algorithm and \cite{RV13}
for a randomized  $\tilde{O}(m\sqrt{n})$ time algorithm). For the sparse graphs, this is an $o(n^2)$ time approximation algorithm. Furthermore, under
plausible assumptions, no $O(n^{2-\epsilon})$ time algorithm can exist that $(3/2-\epsilon')$-approximates (for $\epsilon,\epsilon'> 0$) the diameter \cite{RV13}
and the radius \cite{AWV16} in sparse graphs. Similar results are known also for all eccentricities: a 5/3-approximation to the  eccentricities of all
vertices can be computed in $\tilde{O}(m^{3/2})$ time \cite{ChLaRoScTaVW2014} and, under plausible assumptions, no $O(n^{2-\epsilon})$ time algorithm
can exist that $(5/3-\epsilon')$-approximates (for $\epsilon,\epsilon'> 0$) the  eccentricities of all vertices in sparse graphs \cite{AWV16}. Better
approximation algorithms are known for some special classes of graphs \cite{BrDrNi97,ChDr2003,Chepoi08,CoDrHaPa,CoDrKo2003,Dr1999,DrKo2017,DrNiBr97,WeYu2016}. A number of heuristics for approximating diameters, radii and eccentricities in real-world graphs were proposed and investigated in \cite{AADr,BCH14,BCHKMT15,BCT17,Brandes,CGR16,DrHaLo2018}.

Approximability of APSP is also extensively investigated. An additive $2$-approximation for APSP in unweighted undirected graphs (the graphs we consider in this paper)
was presented in \cite{DoHaZw2000}. It runs in $\tilde{O}(\min\{n^{3/2}m^{1/2}, n^{7/3}\})$ time and hence improves the runtime of an earlier algorithm from \cite{ACIM99}. In \cite{BeKa}, an $\tilde{O}(n^2)$ time algorithm was designed which computes an approximation of all distances with a multiplicative error of 2 and an additive
error of 1. Furthermore, \cite{BeKa} gives an $O(n^{2.24+o(1)}\epsilon^{-3}\log(n/\epsilon))$ time algorithm that computes an approximation of
all distances with a multiplicative error of $(1+\epsilon)$ and an additive
error of 2. The latter improves an earlier algorithm from \cite{El20015}.
Better algorithms are known for some special classes of graphs (see \cite{BrChDr98,Chepoi08,Dragan05,Thorup} and papers cited therein).

The need for fast  approximation algorithms for estimating diameters, radii, centrality indices, or all pairs shortest paths in large-scale complex networks dictates to look for
geometric and topological properties of those networks and utilize them algorithmically. The classical relationships between the diameter, radius, and center of trees and folklore
linear time algorithms for their computation is one of the departing points of this research. A result from 1869 by C. Jordan \cite{Jo} asserts that the radius  of a
tree $T$ is roughly equal to half of its diameter and the center is either the middle vertex or the middle edge of any diametral path. The diameter and a diametral pair of $T$
can be computed (in linear time) by a simple but elegant procedure: pick any vertex $x$, find any vertex $y$ furthest from $x$, and find once more a vertex $z$ furthest
from $y$; then return $\{ y,z\}$ as a diametral pair. One computation of a furthest vertex is called an {\it FP scan}; hence the diameter of a tree can be computed via two FP scans.
This  {\it two FP scans} procedure can be extended to exact or approximate computation of the diameter and radius in many classes of
tree-like graphs. For example, this approach was used to compute the radius and a central vertex of a chordal graph in linear time \cite{ChDr1994}. In this case, the
center of $G$ is still close to the middle of all $(y,z)$-shortest paths and
$d_G(y,z)$ is not the diameter but is still its good approximation: $d(y,z)\ge diam(G)-2$. Even better, the diameter of any chordal graph can be approximated in linear
time with an additive error 1 ~\cite{DrNiBr97}.  But it turns out that the exact computation of diameters of chordal graphs is as difficult as the general diameter problem: it is even difficult
to decide if the diameter of a split graph is 2 or 3.

The experience with chordal graphs shows that one have to abandon the hope of having fast exact algorithms, even for very simple (from metric point of view) graph-classes,
and to search for fast algorithms approximating $diam(G), rad(G), C(G), ecc_G(v)$  with a small additive constant depending only of the coarse geometry of
the graph. {\em Gromov hyperbolicity} or the  {\em negative curvature} of a graph (and, more generally, of a metric space) is one  such constant.
A graph $G=(V,E)$  is $\delta$-{\it hyperbolic} \cite{AlBrCoFeLuMiShSh,GhHa,BrHa,Gr} if for any four vertices $w,v,x,y$ of $G$, the two largest of the three
distance sums $d(w,v)+d(x,y)$, $d(w,x)+d(v,y)$, $d(w,y)+d(v,x)$ differ by at most $2\delta \geq 0$. The {\em hyperbolicity} $\delta(G)$ of a graph $G$ is
the smallest number $\delta$ such that $G$ is $\delta$-hyperbolic. The hyperbolicity can be viewed as a local measure of how close a graph
is metrically to a tree: the smaller the hyperbolicity is, the closer its metric is to a tree-metric (trees are 0-hyperbolic and chordal graphs are 1-hyperbolic).

Recent empirical studies showed that many real-world graphs (including Internet application networks, web networks, collaboration networks, social networks, biological networks,
and others) are tree-like from a metric point of view \cite{AADr,AdcockSM13,Bo++} or have small hyperbolicity \cite{KeSN16,NaSa,ShavittT08}.
It has been suggested in~\cite{NaSa}, and recently formally proved in~\cite{ChDrVa17},  that the property, observed in real-world networks,
in which traffic between nodes tends to go through a relatively small core of the network, as if the shortest paths between them are curved
inwards, is due to the hyperbolicity of the network. Bending property of the eccentricity function in hyperbolic graphs were used in \cite{HendDr,Hend} to identify core-periphery structures in biological networks. Small hyperbolicity in real-world graphs provides also many algorithmic advantages.
Efficient approximate solutions are attainable for a number of optimization problems~\cite{Chepoi08,Chepoi08ENDM,Chepoi12,ChDrVa17,Chepoi07,DGKMY2015,EKS2016,VS2014}.

In \cite{Chepoi08} we initiated the investigation of diameter, center, and radius problems for $\delta$-hyperbolic graphs and we showed that the existing approach for trees
can be extended to this general framework. Namely, it is shown in \cite{Chepoi08} that if $G$ is a $\delta$-hyperbolic graph and $\{ y,z\}$ is the pair returned after two FP scans,
then $d(y,z)\ge diam(G)-2\delta$, $diam(G)\ge 2rad(G)-4\delta-1$, $diam(C(G))\le 4\delta+1$, and $C(G)$ is contained
in a small ball centered at a middle vertex of any shortest $(y,z)$-path. Consequently, we obtained linear time algorithms for the diameter and radius problems  with
additive errors linearly depending on the input graph's hyperbolicity.

In this paper, we advance this line of research and provide a linear time algorithm for approximate computation of the eccentricities (and thus of centrality indices) of all vertices
of a $\delta$-hyperbolic graph $G$, i.e., we compute the approximate values of {\it all eccentricities} within the same time bounds as one computes the approximation of {\it the largest} or
{\it the smallest  eccentricity} ($diam(G)$ or $rad(G)$). Namely, the algorithm outputs for every vertex $v$ of $G$ an estimate $\hat{e}(v)$ of $ecc_G(v)$ such that $ecc_G(v)\leq \hat{e}(v)\leq ecc_G(v)+ c\delta,$ 
where $c>0$ is a small constant. In fact, we demonstrate that $G$ has a shortest path tree, constructible in linear time,
such that for every vertex $v$ of $G$, $ecc_G(v)\leq ecc_T(v)\leq ecc_G(v)+ c\delta$ (a so-called {\em eccentricity $c\delta$-approximating spanning tree}). This is our first main result of this paper and the main ingredient in proving it is the following interesting dependency between the
eccentricities of vertices of $G$ and their distances to the center $C(G)$: up to an additive error linearly depending on $\delta$, $ecc_G(v)$ is equal to $d(v,C(G))$ plus $rad(G)$.
To establish this new result, we have to revisit the results of \cite{Chepoi08} about diameters, radii, and centers,  by simplifying  their proofs and extending them to all eccentricities. 

Eccentricity $k$-approximating spanning trees were introduced by Prisner in~\cite{Prisner}. 
A spanning tree $T$ of a graph $G$ is called an {\em eccentricity $k$-approximating spanning tree} if for every vertex $v$ of $G$  $ecc_T(v)\leq ecc_G(v)+ k$ holds~\cite{Prisner}. Prisner observed that any graph admitting an additive tree $k$-spanner (that is, a spanning tree $T$ such that $d_T(v,u)\leq d_G(v,u)+ k$ for every pair $u,v$) admits also an eccentricity $k$-approximating spanning tree. Therefore, eccentricity $k$-approximating spanning trees exist in interval graphs for $k=2$~\cite{KLM,MVR1996,Prisner1997}, in asteroidal-triple--free graph~\cite{KLM}, strongly chordal
graphs~\cite{BCD1999} and dually chordal graphs~\cite{BCD1999} for
$k=3$. On the other hand, although for every $k$ there is a chordal
graph without an additive  tree $k$-spanner~\cite{KLM,Prisner1997}, yet as
Prisner demonstrated in~\cite{Prisner}, every chordal graph has an
eccentricity 2-approximating spanning tree. Later this result was extended in~\cite{DrKo2017} to a larger family of graphs which includes all chordal graphs and all plane triangulations with inner vertices of degree at least 7. Both those classes belong to the class of 1-hyperbolic graphs. Thus, our result extends the result of~\cite{Prisner} to all $\delta$-hyperbolic graphs.
 
As our second main result, we show that in every $\delta$-hyperbolic graph $G$ all distances with an additive one-sided error of at most $c'\delta$ can be found in $O(|V|^2\log^2|V|)$ time, where $c'< c$ is a small constant. With a recent result in~\cite{fast-appr-hyp}, this demonstrates an equivalence between approximating the hyperbolicity and approximating the distances in graphs. Note that every $\delta$-hyperbolic graph $G$ admits a distance approximating tree $T$~\cite{Chepoi08,Chepoi08ENDM,Chepoi12}, that is, a tree $T$ (which is not necessarily a spanning tree) such that $d_T(v,u)\leq d_G(v,u)+ O(\delta\log n)$ for every pair $u,v$. Such a tree can be used to compute all distances in $G$ with an additive one-sided error of at most $O(\delta\log n)$ in $O(|V|^2)$ time. Our new result removes the dependency of the additive error from $\log n$ and has a much smaller constant in front of $\delta$. Note also that the tree $T$ may use edges not present in $G$ (not a spanning tree of $G$) and thus cannot serve as an eccentricity $O(\delta\log n)$-approximating spanning tree. Furthermore, as chordal graphs are 1-hyperbolic, for every $k$ there is a 1-hyperbolic graph without an additive tree $k$-spanner~\cite{KLM,Prisner1997}.

At the conclusion of this paper, we analyze the performance of our algorithms for approximating eccentricities and distances on a number of real-world networks. Our experimental results show that the estimates on eccentricities and distances obtained are even better than the theoretical bounds proved.

%
%

\section{Preliminaries}

\subsection{Center, diameter, centrality}
All graphs $G=(V,E)$ occurring in this paper are finite, undirected, connected, without loops or multiple edges. We use $n$ and $|V|$ interchangeably to denote the number of vertices and $m$ and $|E|$
to denote the number of edges in $G$.
The {\em length of a path} from a vertex $v$ to a vertex $u$ is the number of edges in the path. The {\em distance} $d_G(u,v)$ between vertices $u$ and $v$ is the length of a shortest path connecting $u$ and $v$ in $G$.
The \emph{eccentricity} of a vertex $v$, denoted by $ecc_G(v)$, is the largest distance from $v$ to any other vertex, i.e., $ecc_G(v)=\max_{u \in V} d_G(v,u)$.  The {\it centrality index} of $v$ is $\frac{1}{ecc_G(v)}$.
The \emph{radius} $rad(G)$ of a graph $G$ is the minimum eccentricity of a vertex in $G$, i.e., $rad(G)=\min_{v \in V} ecc_G(v)$. The \emph{diameter} $diam(G)$ of a graph $G$ is the the maximum eccentricity of a vertex in $G$, i.e., $diam(G)=\max_{v \in V} ecc_G(v)$.
The \emph{center} $C(G)=\{c \in V: ecc_G(c)=rad(G)\}$ of a graph $G$ is the set of vertices with minimum eccentricity.

\subsection{Gromov hyperbolicity and thin geodesic triangles}\label{sec:notions}
Let $(X,d)$ be a metric space. The {\it Gromov product} of $y,z\in X$ with respect to $w$ is defined to be
$$(y|z)_w=\frac{1}{2}(d(y,w)+d(z,w)-d(y,z)).$$
A metric space $(X,d)$ is said to be $\delta$-{\it hyperbolic} \cite{Gr} for $\delta\ge 0$ if
$$(x|y)_w\ge \min \{ (x|z)_w, (y|z)_w\}-\delta$$
for all $w,x,y,z\in X$. Equivalently, $(X,d)$ is $\delta$-hyperbolic
if  for any four points $u,v,x,y$ of $X$, the two largest of the three distance sums
$d(u,v)+d(x,y)$, $d(u,x)+d(v,y)$, $d(u,y)+d(v,x)$ differ by at most
$2\delta \geq 0$. A connected graph $G=(V,E)$  is
$\delta$-{\it hyperbolic} (or of {\it hyperbolicity} $\delta$) if the metric space $(V,d_G)$ is $\delta$-hyperbolic,
where $d_G$ is the standard shortest path metric defined on $G$.

$\delta$-Hyperbolic graphs generalize $k$-chordal graphs and graphs of bounded tree-length:
each $k$-chordal graph has the tree-length at most $\lfloor\frac{k}{2}\rfloor$
\cite{DoGa} and each tree-length $\lambda$ graph has hyperbolicity at most
$\lambda$ \cite{Chepoi08,Chepoi08ENDM}. Recall that a graph is
\emph{$k$-chordal}  if its induced
cycles are of length at most~$k$, and it is of {\it tree-length}
$\lambda$ if it has a Robertson-Seymour tree-decomposition into bags
of diameter at most $\lambda$ \cite{DoGa}.

For geodesic metric spaces and graphs there exist several equivalent
definitions of $\delta$-hyperbolicity involving different but
comparable values of $\delta$ \cite{AlBrCoFeLuMiShSh,BrHa,GhHa,Gr}.
{\it In this paper, we will use the definition via
thin geodesic triangles.} Let $(X,d)$ be a metric space.  A {\it geodesic} joining two
points $x$ and $y$ from $X$ is a (continuous) map $f$ from the segment $[a,b]$
of ${\mathbb R}^1$ of length $|a-b|=d(x,y)$ to $X$ such that
$f(a)=x, f(b)=y,$ and $d(f(s),f(t))=|s-t|$ for all $s,t\in
[a,b].$ A metric space $(X,d)$ is {\it geodesic} if every pair of
points in $X$ can be joined by a geodesic.  Every unweighted graph $G=(V,E)$
equipped with its standard distance $d_G$ can be transformed into a
geodesic (network-like) space $(X,d)$ by replacing every edge
$e=uv$ by a segment $[u,v]$ of length 1; the segments may
intersect only at common ends. Then $(V,d_G)$ is isometrically
embedded in a natural way in $(X,d).$ The
restrictions of geodesics of $X$ to the  vertices $V$ of $G$
are the shortest paths of $G$.

Let $(X,d)$ be a geodesic metric space.  A \textit{geodesic triangle}
$\Delta(x,y,z)$ with $x, y, z \in X$ is the union $[x,y] \cup [x,z]
\cup [y,z]$ of three geodesic segments connecting these vertices. Let $m_x$ be the point of the
geodesic segment $[y,z]$ located at distance $\alpha_y :=(x|z)_y=
(d(y,x)+d(y,z)-d(x,z))/2$ from $y.$ Then $m_x$ is located at
distance $\alpha_z :=(y|x)_z= (d(z,y)+d(z,x)-d(y,x))/2$ from $z$ because
$\alpha_y + \alpha_z = d(y,z)$. Analogously, define the points
$m_y\in [x,z]$ and $m_z\in [x,y]$ both located at distance $\alpha_x
:=(y|z)_x= (d(x,y)+d(x,z)-d(y,z))/2$ from $x;$ see Fig.~\ref{fig1} for an
illustration. There exists a unique isometry $\varphi$ which maps
$\Delta(x,y,z)$ to a tripod $T(x,y,z)$ consisting of three
solid segments $[x,m],[y,m],$ and $[z,m]$ of lengths
$\alpha_x,\alpha_y,$ and $\alpha_z,$ respectively. This isometry
maps the vertices $x,y,z$ of $\Delta(x,y,z)$ to the respective
leaves of $T(x,y,z)$ and the points $m_x,m_y,$
and $m_z$ to the center $m$ of this tripod. Any other point  of
$T(x,y,z)$ is the image of exactly two points of $\Delta
(x,y,z).$ A geodesic triangle $\Delta(x,y,z)$ is called
$\delta$-{\it thin} if for all points $u,v\in \Delta(x,y,z),$
$\varphi(u)=\varphi(v)$ implies $d(u,v)\le \delta.$ A graph $G=(V,E)$
whose all geodesic triangles $\Delta(u,v,w)$, $u,v,w\in V$, are $\delta$-thin
is called a {\em graph with $\delta$-thin triangles}, and $\delta$ is called the {\em thinness} parameter of $G$.

\begin{figure}
\begin{center}
\vspace*{-0.4cm}
\includegraphics*[width=8truecm]{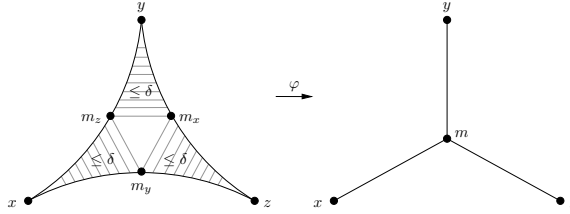}
\end{center}
\caption{A geodesic triangle $\Delta(x,y,z),$ the points $m_x, m_y,
m_z,$ and the tripod $\Upsilon(x,y,z)$} 
\label{fig1}
\vspace*{-0.2cm}
\end{figure}

The following result shows that hyperbolicity of a
geodesic space or a graph is equivalent to having thin geodesic
triangles.

\begin{proposition} [\cite{AlBrCoFeLuMiShSh,BrHa,GhHa,Gr}]\label{hyp_charact}
Geodesic triangles of geodesic $\delta$-hyperbolic spaces or graphs are $4\delta$-thin.
Conversely, geodesic spaces or graphs with $\delta$-thin triangles are $\delta$-hyperbolic.
\end{proposition}

In what follows, we will need few more notions and notations. Let $G=(V,E)$ be a graph.
By $[x,y]$ we denote a shortest path connecting vertices $x$ and $y$ in $G$; we call $[x,y]$ a {\it geodesic}
between $x$ and $y$. A {\em ball} $B(s,r)$ of $G$ centered at vertex $s \in V$ and with radius $r$ is the set of
all vertices with distance no more than $r$ from $s$ (i.e., $B(s,r):=\{v\in V: d_G(v,s) \leq r \}$). The {\em $k$th-power}
of a graph $G=(V,E)$ is the graph $G^k=(V,E')$ such that $xy\in E'$ if and only if $0<d_G(x,y)\le k$.
Denote by  $F(x):=\{y\in V: d_G(x,y)=ecc_G(x)\}$ the set of all vertices of $G$ that are {\em most distant} from $x$. Vertices $x$ and $y$ of $G$ are called {\em mutually distant} if $x\in F(y)$ and $y\in F(x)$, i.e., $ecc_G(x)=ecc_G(y)=d_G(x,y)$.

\section{Fast approximation of eccentricities}\label{sec:ecc}
In this section, we give linear and almost linear time algorithms for sharp estimation of the diameters, the radii, the centers and
the eccentricities of all vertices in graphs  with $\delta$-thin triangles. Before presenting those algorithms, we establish some conditional lower
bounds on complexities of computing the diameters and the radii in those graphs.

\subsection{Conditional lower bounds on complexities}

Recent work has revealed convincing evidence that solving the diameter problem in subquadratic time might
not be possible, even in very special classes of graphs. Roditty and Vassilevska W. \cite{RV13} showed that an algorithm that
can distinguish between diameter 2 and 3 in a sparse graph in subquadratic time refutes the
following widely believed conjecture.
\medskip

{\em The Orthogonal Vectors Conjecture:} There is no $\epsilon > 0$ such that for all $c\geq 1$, there is  an algorithm that given two lists of  $n$ binary vectors $A,B\subseteq \{0, 1\}^d$  where $d = c \log n$ can determine if there is an
orthogonal pair $a\in A, b\in B$, in $O(n^{2-e})$ time. \medskip

Williams \cite{W04} showed that the Orthogonal Vectors (OV) Conjecture is implied by the well-known Strong Exponential Time Hypothesis (SETH) of Impagliazzo, Paturi, and Zane \cite{IPZ01,IP01}. Nowadays many papers base the hardness of problems on SETH and the OV
conjecture (see, e.g.,  \cite{AWV16,BCH14,V-W15} and papers cited therein).

Since all geodesic triangles of a graph constructed in the reduction in \cite{RV13} are 2-thin, we can rephrase the result from \cite{RV13} as follows.

\begin{statement}\label{stat:lower-diam} If for some $\epsilon>0$, there is an algorithm that can determine if a given graph with 2-thin triangles, $n$ vertices and $m=O(n)$ edges has diameter 2 or 3 in $O(n^{2-\epsilon})$ time, then the Orthogonal Vector Conjecture is false.
\end{statement}

To prove a similar lower bound result for the radius problem, recently Abboud et al. \cite{AWV16} suggested to use the following natural
and plausible variant of the OV conjecture.
\medskip

{\em The Hitting Set Conjecture:} There is no $\epsilon > 0$ such that for all $c\geq 1$, there is  an algorithm that given two lists $A,B$
of  $n$ subsets of a universe $U$ of size $c \log n$, can decide
in $O(n^{2-e})$  time if there is a set in the first list that intersects
every set in the second list, i.e. a hitting set.
\medskip

Abboud et al. \cite{AWV16} showed that an algorithm that
can distinguish between radius 2 and 3 in a sparse graph 
in subquadratic time refutes the
Hitting Set Conjecture. Since all geodesic triangles of a graph constructed in the reduction in \cite{AWV16} are 2-thin,  rephrasing that result from \cite{AWV16}, we have.

\begin{statement}\label{stat:lower-rad} If for some $\epsilon>0$, there is an algorithm that can determine if a given graph with 2-thin triangles, $n$ vertices, and $m=O(n)$ edges has radius 2 or 3 in $O(n^{2-\epsilon})$ time, then the Hitting Set Conjecture is false.
\end{statement}


\subsection{Fast additive approximations}

In this subsection, we show that in a graph $G$ with $\delta$-thin triangles the eccentricities of all vertices can be computed in total linear time with an additive error depending on $\delta$. We establish that the eccentricity of a vertex is determined (up-to a small error) by how far the vertex is from the center $C(G)$ of $G$. Finally, we show how to construct a spanning tree $T$ of $G$ in which the eccentricity of any vertex is its eccentricity in $G$ up to an additive error depending only on $\delta$. For these purposes, we revisit and extend several results from our previous paper \cite{Chepoi08}  concerning the linear time approximation of diameter, radius, and centers of $\delta$-hyperbolic graphs. For these particular cases, we provide simplified proofs, leading to better additive errors due to the use of thinness of triangles instead of the four point condition and to the computation in $O(\delta|E|)$ time of a pair of mutually distant vertices.

Define the eccentricity layers of a graph $G$ as follows: for $k=0,\dots,diam(G)-rad(G)$ set
$$C^k(G):=\{v\in V: ecc_G(v)=rad(G)+k\}.$$ With this notation, the center of a graph is  $C(G)=C^{0}(G)$. In what follows, it will be convenient to define also the eccentricity of the middle point $m$ of any edge $xy$ of $G$; set $ecc_G(m)=\min\{ecc_G(x),ecc_G(y)\}+1/2$.

We start with a proposition showing that, in a graph $G$ with $\delta$-thin triangles, a middle vertex of any geodesic between two mutually distant vertices has the eccentricity close to $rad(G)$ and is not too far from the center $C(G)$ of $G$.

\begin{proposition}\label{prop:middlepoint-ecc} Let $G$ be a graph with $\delta$-thin triangles, $u,v$ be a pair of mutually distant vertices of $G$.
 \begin{enumerate}
   \item[$(a)$] If $c^*$ is the middle point of any $(u,v)$-geodesic, then $ecc_G(c^*)\leq \frac{d_G(u,v)}{2} +\delta\leq rad(G)+\delta$.
   \item[$(b)$] If $c$ is a middle vertex of any $(u,v)$-geodesic, then $ecc_G(c)\leq \lceil{\frac{d_G(u,v)}{2}}\rceil +\delta\leq rad(G)+\delta$.
   \item[$(c)$] $d_G(u,v)\geq  2rad(G)-2\delta-1$. In particular, $diam(G)\geq 2rad(G)-2\delta-1.$
\item[$(d)$] If $c$ is a middle vertex of any $(u,v)$-geodesic and $x\in C^k(G),$ then $k-\delta \le d_G(x,c) \le k+2\delta+1$. In particular, \mbox{$C(G)\subseteq B(c,2\delta+1)$}.
 \end{enumerate}
\end{proposition}

\begin{proof} Let $x$ be an arbitrary vertex of $G$ and $\Delta(u,v,x):=[u,v]\cup [v,x]\cup [x,u]$ be a geodesic triangle, where $[v,x], [x,u]$ are arbitrary geodesics connecting $x$ with $v$ and $u$.
Let $m_x$ be a point on $[u,v]$ which is at distance $(x|u)_v= \frac{1}{2}(d(x,v)+d(v,u)-d(x,u))$ from $v$ and hence at distance $(x|v)_u= \frac{1}{2}(d(x,u)+d(v,u)-d(x,v))$ from $u$. Since $u$ and $v$ are mutually distant, we can assume, without loss of generality, that $c^*$  is located on $[u,v]$ between $v$ and $m_x$, i.e., $d(v,c^*)\leq d(v,m_x)=(x|u)_v$, and hence $(x|v)_u\leq (x|u)_v$.
Since $d_G(v,x)\leq d_G(v,u)$, we also get $(u|v)_x\leq (x|v)_u$.

(a) By the triangle inequality and since $d_G(u,v)\leq diam(G)\leq 2rad(G)$, we get
\begin{eqnarray}
\nonumber d_G(x,c^*)& \leq & e(u|v)_x+\delta+d_G(u,c^*)-(x|v)_u\\
\nonumber               & \leq & d_G(u,c^*)+\delta= \frac{d_G(u,v)}{2} +\delta\leq rad(G)+\delta.
\end{eqnarray}

(b) 
Since $c^*=c$ when $d_G(u,v)$ is even and $d_G(c^*,c)=\frac{1}{2}$ when $d_G(u,v)$ is odd, we have \mbox{$ecc_G(c)\leq ecc_G(c^*)+\frac{1}{2}$}. Additionally to the proof of (a), one needs only to consider the case when $d_G(u,v)$ is odd. We know that  the middle point $c^*$ sees all vertices of $G$ within distance at most $\frac{d_G(u,v)}{2}+\delta$. Hence, both ends of the edge of  $(u,v)$-geodesic,  containing the point $c^*$ in the middle, have eccentricities at most $$\frac{d_G(u,v)}{2}+ \frac{1}{2}+\delta=\lceil{\frac{d_G(u,v)}{2}}\rceil +\delta\leq \lceil{\frac{2rad(G)-1}{2}}\rceil+\delta=rad(G)+\delta.$$

(c) Since a middle vertex $c$ of any $(u,v)$-geodesic sees all vertices of $G$ within distance at most \mbox{$\lceil{\frac{d_G(u,v)}{2}}\rceil +\delta$}, if $d_G(u,v)\leq 2rad(G)-2\delta-2$, then
$$ecc_G(c)\leq \lceil{\frac{d_G(u,v)}{2}}\rceil +\delta\leq \lceil{\frac{2rad(G)-2\delta-2}{2}}\rceil +\delta < rad(G),$$ which is impossible.

(d) In the proof of (a), instead of an arbitrary vertex $x$, consider any vertex $x$ from $C^k(G)$. By the triangle inequality and since $d_G(u,v)\geq 2rad(G)-2\delta-1$ and both $d_G(u,x),d_G(x,v)$ are at most $rad(G)+k$,  we get
\begin{eqnarray}
\nonumber d_G(x,c^*)& \leq & (u|v)_x+\delta+(x|u)_v-d_G(v,c^*)= d_G(v,x)-d_G(v,c^*)+\delta  \\
\nonumber               & \leq & rad(G)+k-\frac{d_G(u,v)}{2} +\delta\leq k+2\delta+\frac{1}{2}.
\end{eqnarray}
Consequently, $d_G(x,c)\leq d_G(x,c^*)+\frac{1}{2}\leq k+2\delta+1.$ On the other hand, since $ecc_G(x)\le ecc_G(c)+d_G(x,c)$ and  $ecc_G(c)\le rad(G)+\delta$, by statement (a), we get
\begin{eqnarray}
\nonumber d_G(x,c)& \geq & ecc_G(x) - ecc_G(c) = k+rad(G)-ecc_G(c)\\
\nonumber               & \geq & (k+rad(G)) - (rad(G)+\delta) = k-\delta.
\end{eqnarray}
\qed
\end{proof}

As an easy consequence of Proposition \ref{prop:middlepoint-ecc}(d), we get that the eccentricity $ecc_G(x)$ of any vertex $x$ is equal, up to an additive one-sided error of at most $4\delta+2$, to $d_G(x,C(G))$ plus $rad(G)$.

\begin{corollary}\label{prop:formula-ecc}
For every vertex $x$ of a graph $G$ with $\delta$-thin triangles,
$$d_G(x,C(G))+rad(G)-4\delta-2 \le ecc_G(x)\le d_G(x,C(G))+rad(G).$$
\end{corollary}

\begin{proof}
Consider an arbitrary vertex $x$ in $G$ and assume that $ecc_G(x)=rad(G)+k$. Let $c_x$ be a vertex from $C(G)$ closest to $x$.
By Proposition \ref{prop:middlepoint-ecc}(d), $d_G(c,c_x)\leq 2\delta +1$ and $d_G(x,c)\le k+2\delta+1=ecc_G(x)-rad(G)+2\delta+1$.  Hence, $$d_G(x,C(G))=d_G(x, c_x)\leq d_G(x,c)+d_G(c, c_x)\le d_G(x,c)+2\delta+1$$ and $$ecc_G(x)\ge d_G(x,c)+rad(G)-2\delta-1.$$
Combining both inequalities, we get $$ecc_G(x)\ge d_G(x,C(G))+rad(G)-4\delta-2.$$
Note also that, by the triangle inequality, $ecc_G(x)\leq d_G(x,c_x)+ecc_G(c_x)= d_G(x,C(G))+rad(G)$ (that is, the right-hand inequality holds for all graphs). \qed
\end{proof}

It is interesting to note that the equality $ecc_G(x)= d_G(x,C(G))+rad(G)$ holds for every vertex of a graph $G$ if and only if the eccentricity function $ecc_G(\cdot)$ on $G$ is unimodal (that is, every local minimum is a global minimum)\cite{Dragan-Diss}. A slightly weaker condition holds for all chordal graphs  \cite{DrKo2017}: for every vertex $x$ of a chordal graph $G$, $ecc_G(x)\ge d_G(x,C(G))+rad(G)-1$.


\begin{proposition}\label{prop:r-vertex-ecc} Let $G$ be a graph with $\delta$-thin triangles and $u,v$ be a pair of vertices of $G$ such that $v\in F(u)$.
 \begin{enumerate}
   \item[$(a)$] If $w$ is a vertex of a $(u,v)$-geodesic at distance $rad(G)$ from $v$, then $ecc_G(w)\leq rad(G)+\delta$.
   \item[$(b)$] For every pair of vertices $x,y\in V$, $\max\{d_G(v,x),d_G(v,y)\}\geq d_G(x,y)-2\delta$.
   \item[$(c)$] $ecc_G(v)\geq diam(G)-2\delta\geq 2rad(G)-4\delta-1$.
  \item[$(d)$] If $t\in F(v),$ $c$ is a vertex of a $(v,t)$-geodesic at distance $\lceil\frac{d_G(v,t)}{2}\rceil$ from $t$ and $x\in C^k(G)$, then $ecc_G(c)\leq rad(G)+3\delta$ and $k-3\delta \le d_G(x,c) \le k+3\delta+1$. In particular, $C(G)\subseteq B(c,3\delta+1)$.
 \end{enumerate}
\end{proposition}

\begin{proof} (a) Let $x$ be a vertex of $G$ with $d_G(w,x)=ecc_G(w)$. Let $\Delta(u,v,x):=[u,v]\cup [v,x]\cup [x,u]$ be a geodesic triangle, where $[v,x], [x,u]$ are arbitrary geodesics connecting $x$ with $v$ and $u$.
Let $m_x$ be a point on $[u,v]$ which is at distance $(x|u)_v= \frac{1}{2}(d(x,v)+d(v,u)-d(x,u))$ from $v$ and hence at distance \mbox{$(x|v)_u= \frac{1}{2}(d(x,u)+d(v,u)-d(x,v))$} from $u$. We distinguish between two cases: $w$ is between $u$ and $m_x$ or  $w$ is between $v$ and $m_x$ in $[u,v]$.

In the first case, by the triangle inequality and  $d_G(u,x)\leq d_G(u,v)$ (and hence, $(u|x)_v\geq (u|v)_x$), we get
$$d_G(w,x)\leq rad(G)-(u|x)_v+\delta+(u|v)_x\leq rad(G)+\delta.$$

In the second case, by the triangle inequality and since $d_G(v,x)\leq diam(G)\leq 2rad(G)$, we get
\begin{eqnarray}
\nonumber d_G(w,x)& \leq & (u|x)_v-rad(G)+\delta+(u|v)_x \\
\nonumber                      & \leq & d_G(x,v)-rad(G)+\delta \\
\nonumber                      & \leq & 2rad(G)-rad(G)+\delta = rad(G)+\delta.
\end{eqnarray}

(b) Consider an arbitrary $(u,v)$-geodesic $[u,v]$. Let $\Delta(u,v,x):=[u,v]\cup [v,x]\cup [x,u]$ be a geodesic triangle, where $[v,x], [x,u]$ are arbitrary geodesics connecting $x$ with $v$ and $u$. Let $\Delta(u,v,y):=[u,v]\cup [v,y]\cup [y,u]$ be a geodesic triangle, where $[v,y], [y,u]$ are arbitrary geodesics connecting $y$ with $v$ and $u$.

Let $m_x$ be a point on $[u,v]$ which is at distance $(x|u)_v= \frac{1}{2}(d(x,v)+d(v,u)-d(x,u))$ from $v$ and hence at distance \mbox{$(x|v)_u= \frac{1}{2}(d(x,u)+d(v,u)-d(x,v))$} from $u$. Let $m_y$ be a point on $[u,v]$ which is at distance $(y|u)_v= \frac{1}{2}(d(y,v)+d(v,u)-d(y,u))$ from $v$ and hence at distance \mbox{$(y|v)_u= \frac{1}{2}(d(y,u)+d(v,u)-d(y,v))$} from $u$. Without loss of generality, assume that $m_x$ is on $[u,v]$ between $v$ and $m_y$.

Since $d_G(u,v)\geq d_G(u,x)$ (as $v\in F(u)$), we have $(u|v)_x\leq (u|x)_v$. By the triangle inequality,
we get
\begin{eqnarray}
\nonumber d_G(x,y)& \leq & (u|v)_x+\delta+((y|u)_v- (u|x)_v)+\delta+(u|v)_y\\
\nonumber                      & \leq & (u|x)_v-(u|x)_v+2\delta+(y|u)_v +(u|v)_y\\
\nonumber                      & = & d_G(v,y)+2\delta.
\end{eqnarray}

Consequently, $\max\{d_G(v,x),d_G(v,y)\}\geq d_G(v,y)\geq d_G(x,y)-2\delta.$

(c) Now, if $x,y$ is a diametral pair, i.e., $d_G(x,y)=diam(G)$, then, by (b) and Proposition \ref{prop:middlepoint-ecc}(c),
\begin{eqnarray}
\nonumber ecc_G(v)& \geq & \max\{d_G(v,x),d_G(v,y)\} \\
\nonumber                      & \geq & d_G(x,y)-2\delta = diam(G)-2\delta\\
\nonumber                      & \geq & 2rad(G)-4\delta-1.
\end{eqnarray}

(d) Consider any $(v,t)$-geodesic $[v,t]$ and let $c^*$ be the middle point of it, $w$ be a vertex of $[v,t]$  at distance $rad(G)$ from $t$, and
$c$ be a vertex of $[v,t]$ at distance $\lceil\frac{d_G(v,t)}{2}\rceil$ from $t$.
We know by (a) that $ecc_G(w)\leq rad(G)+\delta$. Furthermore, since $2rad(G)\geq d_G(v,t)\geq  2rad(G)-4\delta-1$ (by (c)), \mbox{$rad(G)\geq d_G(t,c)=\lceil\frac{d_G(v,t)}{2}\rceil\geq rad(G)-2\delta$.} Hence, $$d_G(w,c)= d_G(w,t)-d_G(c,t)\leq rad(G)- rad(G)+2\delta=2\delta,$$
implying $$ecc_G(c)\leq d_G(w,c)+ecc_G(w)\leq rad(G)+3\delta.$$

Let now $x$ be an arbitrary vertex from  $C^k(G)$, i.e., $ecc_G(x)\leq rad(G)+k$, for some integer $k\geq 0$.  Consider a geodesic triangle $\Delta(t,v,x):=[t,v]\cup [v,x]\cup [x,t]$, where $[v,x], [x,t]$ are arbitrary geodesics connecting $x$ with $v$ and $t$.
Let $m_x$ be a point on $[t,v]$ which is at distance $(x|t)_v= \frac{1}{2}(d(x,v)+d(v,t)-d(x,t))$ from $v$ and hence at distance $(x|v)_t= \frac{1}{2}(d(x,t)+d(v,t)-d(x,v))$ from $t$. Since, in what follows, we will use only the fact that $d_G(v,t)\geq  2rad(G)-4\delta-1$, we can assume, without loss of generality, that $c^*$  is located on $[t,v]$ between $v$ and $m_x$, i.e., $d(v,c^*)\leq d(v,m_x)=(x|t)_v$.

By the triangle inequality and since $d_G(v,t)\geq 2rad(G)-4\delta-1$  and both $d_G(t,x)$ and $d_G(x,v)$ are at most $rad(G)+k$,  we get
\begin{eqnarray}
\nonumber d_G(x,c^*)& \leq & (t|v)_x+\delta+(x|t)_v-d_G(v,c^*) =  d_G(v,x)-d_G(v,c^*)+\delta\\
\nonumber                      & \leq & rad(G)+k-\frac{d_G(v,t)}{2} +\delta\leq  k+3\delta+\frac{1}{2}.
\end{eqnarray}
Hence, $d_G(x,c)\leq d_G(x,c^*)+\frac{1}{2}\leq k+3\delta+1.$ On the other hand, since $ecc_G(x)\le ecc_G(c)+d_G(x,c)$ and $ecc_G(c)\le rad(G)+3\delta$, we get
\begin{eqnarray}
\nonumber d_G(x,c)& \geq & ecc_G(x) - ecc_G(c) = k+rad(G)-ecc_G(c)\\
\nonumber               & \geq & (k+rad(G)) - (rad(G)+3\delta) = k-3\delta.
\end{eqnarray}
\qed
\end{proof}

\begin{proposition} \label{prop:center-diam}
For every graph $G$ with $\delta$-thin triangles, $diam(C^k(G))\leq 2k+2\delta+1.$ In particular, $diam(C(G))\leq 2\delta+1.$
\end{proposition}

\begin{proof} Let $x,y$ be two vertices of $C^k(G)$ such that $d_G(x,y)=diam(C^k(G))$. Pick any $(x,y)$-geodesic and consider the middle point $m$ of it.
Let $z$ be a vertex of $G$ such that $d_G(m,z)=ecc_G(m)$. Consider a geodesic triangle $\Delta(x,y,z):=[x,y]\cup [y,z]\cup [z,x]$, where $[z,x], [y,z]$ are arbitrary geodesics connecting $z$ with $x$ and $y$.
Let $m_z$ be a point on $[x,y]$ which is at distance $(x|z)_y= \frac{1}{2}(d(x,y)+d(z,y)-d(x,z))$ from $y$ and hence at distance $(y|z)_x= \frac{1}{2}(d(x,y)+d(z,x)-d(y,z))$ from $x$. Without loss of generality,  we can assume that $m$  is located on $[x,y]$ between $y$ and $m_z$.

Since $ecc_G(y)\leq rad(G)+k$, we have
$$d_G(m,z)=ecc_G(m)\geq rad(G)- \frac{1}{2}\geq ecc_G(y)- k-\frac{1}{2}\geq d_G(y,z) - k-\frac{1}{2}.$$
On the other hand, by the
triangle inequality, we get
\begin{eqnarray}
\nonumber d_G(m,z)& \leq & (x|z)_y-d_G(y,m)+\delta+(x|y)_z=d_G(y,z)- d_G(y,m)+\delta\\
\nonumber               & \leq & d_G(y,z)-\frac{d_G(x,y)}{2}+\delta.
\end{eqnarray}

Hence, $d_G(x,y)\leq 2k + 2\delta+1.$ \qed
\end{proof}

\subsubsection{Diameter and radius.}

For an arbitrary connected graph $G=(V,E)$ and a given vertex $u\in V$, a most distant from $u$ vertex $v\in F(u)$ can be found in linear ($O(|E|)$) time by a {\em breadth-first-search} $BFS(u)$ started at $u$.
A pair of mutually distant vertices of a connected graph $G=(V,E)$ with $\delta$-thin triangles can be computed in $O(\delta |E|)$ total time as follows. By Proposition \ref{prop:r-vertex-ecc}(c), if $v$ is a most distant vertex from an arbitrary vertex $u$ and $t$ is a most distant vertex from $v$, then $d(v,t)\geq \diam(G)-2\delta$. Hence, using at most $O(\delta)$ {\em breadth-first-searches}, one can generate a sequence of vertices $v:=v_1,t:=v_2, v_3, \dots v_k$ with $k\leq 2\delta+2$  such that each $v_i$ is most distant from $v_{i-1}$ (with, $v_0=u$) and $v_k$, $v_{k-1}$ are mutually distant vertices (the initial value $d(v,t)\geq \diam(G)-2\delta$ can be improved at most $2\delta$ times).

Thus, by Proposition \ref{prop:middlepoint-ecc} and Proposition \ref{prop:r-vertex-ecc}, we get the following additive approximations for the radius and the diameter of a graph with $\delta$-thin triangles.

\begin{corollary} \label{th:rad-diam-appr}
Let $G=(V,E)$ be a graph with $\delta$-thin triangles.
\begin{enumerate}
   \item[1.] There is a linear $(O(|E|))$ time algorithm which finds in $G$ a vertex $c$ with eccentricity at most $rad(G)+3\delta$ and a vertex $v$ with eccentricity at least $diam(G)-2\delta$. Furthermore, $C(G)\subseteq B(c,3\delta+1)$ holds.
   \item[2.] There is an almost linear $(O(\delta|E|))$ time algorithm which finds in $G$ a vertex $c$ with eccentricity at most $rad(G)+\delta$. Furthermore, $C(G)\subseteq B(c,2\delta+1)$ holds.
\end{enumerate}
\end{corollary}

\subsubsection{All eccentricities.}
In what follows, we will show that all vertex eccentricities of a graph with $\delta$-thin triangles can be also additively approximated in (almost) linear time.

\begin{proposition}\label{prop:ecc-appr-tree} Let $G$ be a graph with $\delta$-thin triangles.
 \begin{enumerate}
   \item[$(a)$] If $c$ is a middle vertex of any $(u,v)$-geodesic between a pair $u,v$ of mutually distant vertices of $G$ and $T$ is a $BFS(c)$-tree of $G$, then, for every vertex $x$ of $G$, $ecc_G(x)\leq ecc_T(x)\leq ecc_G(x)+ 3\delta+1.$
   \item[$(b)$] If $v$ is a most distant vertex from an arbitrary vertex $u$, $t$ is a most distant vertex from $v$,
    $c$ is a vertex of a $(v,t)$-geodesic at distance $\lceil\frac{d_G(v,t)}{2}\rceil$ from $t$ and $T$ is a $BFS(c)$-tree of $G$, then  $ecc_G(x)\leq ecc_T(x)\leq ecc_G(x)+ 6\delta+1.$
 \end{enumerate}
\end{proposition}

\begin{proof} (a) Let $x$ be an arbitrary vertex of $G$ and assume that $ecc_G(x)=rad(G)+k$ for some integer $k\ge 0$. We know from Proposition \ref{prop:middlepoint-ecc}(b) that $ecc_G(c)\leq rad(G)+\delta$. Furthermore, by Proposition \ref{prop:middlepoint-ecc}(d), $d_G(c,x)\leq k+2\delta+1$.
Since $T$ is a $BFS(c)$-tree, $d_G(x,c)=d_T(x,c)$ and $ecc_G(c)=ecc_T(c)$.
Consider a vertex $y$ in $G$ such that $d_T(x,y)=ecc_T(x).$ We have
\begin{eqnarray}
\nonumber ecc_T(x)& = & d_T(x,y) \leq  d_T(x,c)+ d_T(c,y) \\
\nonumber                      & \leq & d_G(x,c)+ecc_T(c) =  d_G(x,c)+ecc_G(c)  \\
\nonumber                      & \leq & k+2\delta+1+rad(G)+\delta  =  rad(G)+k+3\delta+1  \\
\nonumber                      & = & ecc_G(x)+3\delta+1.
\end{eqnarray}
As $T$ is a spanning tree of $G$, evidently,  also $ecc_G(x)\leq ecc_T(x)$ holds.

(b) The proof is similar to the proof of (a); only, in this case, $ecc_G(c)\leq rad(G)+3\delta$ and $d_G(c,x)\leq k+3\delta +1$ holds for every $x\in C^k(G)$ (by Proposition \ref{prop:r-vertex-ecc}(d)). \qed
\end{proof}

A spanning tree $T$ of a graph $G$ is called an {\em eccentricity $k$-approximating spanning tree} if for every vertex $v$ of $G$  $ecc_T(v)\leq ecc_G(v)+ k$ holds~\cite{DrKo2017,Prisner}. Thus, by Proposition  \ref{prop:ecc-appr-tree}, we get.

\begin{theorem} \label{th:ecc-appr-tree}
Every graph $G=(V,E)$ with $\delta$-thin triangles admits an eccentricity $(3\delta+1)$-approximating spanning tree constructible in $O(\delta|E|)$ time and an eccentricity $(6\delta+1)$-approximating spanning tree constructible in $O(|E|)$ time.
\end{theorem}
Theorem \ref{th:ecc-appr-tree} generalizes recent results from~\cite{DrKo2017,Prisner} that chordal graphs and some of their generalizations admit eccentricity 2-approximating spanning trees.

Note that the eccentricities of all vertices in any tree $T=(V,U)$ can be computed in $O(|V|)$ total time. As we noticed already, it is a folklore by now that for trees the following facts are true:
\begin{itemize}
   \item[(1)] The center $C(T)$ of any tree $T$ consists of one vertex or two adjacent vertices.
   \item[(2)] The center $C(T)$ and the radius $rad(T)$ of any tree $T$ can be found in linear time.
   \item[(3)] For every vertex $v\in V$, $ecc_T(v)=d_T(v,C(T))+rad(T)$.
 \end{itemize}
Hence, using $BFS(C(T))$ on $T$ one can compute $d_T(v,C(T))$ for all $v\in V$ in total $O(|V|)$ time. Adding now $rad(T)$ to $d_T(v,C(T))$, one gets  $ecc_T(v)$ for all $v\in V$. Consequently, by
 Theorem \ref{th:ecc-appr-tree}, we get the following additive approximations for the vertex eccentricities in graphs with $\delta$-thin triangles.

\begin{theorem} \label{th:all-ecc--appr}
Let $G=(V,E)$ be a graph with $\delta$-thin triangles.
\begin{enumerate}
   \item[(1)] There is an algorithm which in total linear $(O(|E|))$ time outputs for every vertex $v\in V$ an estimate $\hat{e}(v)$ of its eccentricity $ecc_G(v)$ such that $ecc_G(v)\leq \hat{e}(v)\leq ecc_G(v)+ 6\delta+1.$
   \item[(2)] There is an algorithm which in total almost linear $(O(\delta|E|))$ time outputs for every vertex $v\in V$ an estimate $\hat{e}(v)$ of its eccentricity $ecc_G(v)$ such that $ecc_G(v)\leq \hat{e}(v)\leq ecc_G(v)+ 3\delta+1.$
\end{enumerate}
\end{theorem}


\section{Fast Additive Approximation of All Distances} \label{sec:dist}

Here, we will show that if the $\delta$th power $G^{\delta}$ of a graph $G$ with $\delta$-thin triangles is known in advance, then the distances in $G$ can be additively approximated (with an additive one-sided error of at most $\delta+1$) in $O(|V|^2)$ time. If $G^{\delta}$ is not known, then the distances can be additively approximated (with an additive one-sided error of at most $2\delta+2$) in almost quadratic time.

Our method is a generalization of an unified approach
used in \cite{Dragan05} to estimate (or compute exactly)  all pairs shortest paths in such special graph families as $k$-chordal graphs, chordal graphs, AT-free graphs and many others. For example: all distances in $k$-chordal graphs with an additive one-sided error of at most $k - 1$ can be found in $O(|V|^2)$ time;  all distances in chordal graphs with an additive one-sided error of at most 1 can be found in $O(|V|^2)$ time and  the all pairs shortest path problem on a chordal graph $G$ can be solved in $O(|V|^2)$ time if $G^2$ is known. Note that in chordal graph all geodesic triangles are 2-thin.


\medskip

Let $G=(V,E)$ be a graph with $\delta$-thin triangles. Pick an arbitrary start vertex $s\in V$ and construct a {\em $BFS(s)$-tree} $T$ of $G$ {\em rooted at} $s$. Denote by $p_T(x)$ the {\em parent} and by $h_T(x)=d_T(x,s)=d_G(x,s)$ the {\em height} of a vertex $x$ in $T$. Since we will deal only with one tree $T$, we will often omit the subscript $T$. Let $P_T(x,s):=(x_q,x_{q-1},\dots,x_1,s)$ and $P_T(y,s):=(y_p,y_{p-1},\dots,y_1,s)$ be the paths of $T$ connecting vertices $x$ and $y$ with the root $s$. By $sl_T(x,y;\lambda)$ we denote the largest index $k$ such that $d_G(x_k,y_k)\leq \lambda$ (the $\lambda$ \underline{s}eparation \underline{l}evel). Our method is based on the following simple fact.

\begin{proposition} \label{prop:formula}
For every vertices $x$ and $y$ of a graph $G$ with $\delta$-thin triangles and any $BFS$-tree $T$ of $G$, $$h_T(x)+h_T(y)-2k-1\leq d_G(x,y)\leq h_T(x)+h_T(y)-2k+d_G(x_k,y_k),$$ where $k=sl_T(x,y;\delta)$.
\end{proposition}

\begin{proof} By the triangle inequality, $d_G(x,y)\leq d_G(x,x_k)+ d_G(x_k,y_k)+ d_G(y_k,y)= h_T(x)+h_T(y)-2k+d_G(x_k,y_k)$.  Consider now an arbitrary $(x,y)$-geodesic $[x,y]$ in $G$. Let $\Delta(x,y,s):=[x,y]\cup [x,s]\cup [y,s]$ be a geodesic triangle, where $[x,s]=
P_T(x,s)$ and $[y,s]=P_T(y,s)$. Since $\Delta(x,y,s)$ is $\delta$-thin,  $sl_T(x,y;\delta)\ge (x|y)_s-\frac{1}{2}$. Hence, $h_T(x)- sl_T(x,y;\delta)\le (s|y)_x+\frac{1}{2}$ and $h_T(y)- sl_T(x,y;\delta)\le (s|x)_y+\frac{1}{2}$. As $d_G(x,y)=(s|y)_x+(s|x)_y$, we get $d_G(x,y)\ge h_T(x)- sl_T(x,y;\delta) + h_T(y)- sl_T(x,y;\delta)-1$. \qed
\end{proof}

Note that we may regard $BFS(s)$ as having produced a numbering from $n$ to 1 in decreasing order of the vertices in $V$ where vertex $s$ is numbered $n$. As a vertex is placed in the queue by $BFS(s)$, it is given the next available number. The last vertex visited is given the number 1.  Let $\sigma:=[v_1, v_2,\dots, v_n=s]$ be a $BFS(s)$-ordering of the vertices of $G$ and $T$ be a $BFS(s)$-tree of $G$ produced by a $BFS(s)$. Let $\sigma(x)$ be the number assigned to a vertex $x$ in this $BFS(s)$-ordering. For two vertices $x$ and $y$, we write $x<y$ whenever $\sigma(x)<\sigma(y)$.

First, we will show that if $G^{\delta}$ is known in advance (i.e., its adjacency matrix is given) for a graph $G$ with $\delta$-thin triangles, then the distances in $G$ can be additively approximated (with an additive one-sided error of at most $\delta+1$) in $O(|V|^2)$ time.
We consider the vertices of $G$ in the order $\sigma$ from 1 to $n$. For each current vertex $x$ we show that the values $\widehat{d}(x,y):=h_T(x)+h_T(y)-2sl_T(x,y;\delta)+\delta$
for all vertices $y$ with $y>x$ can be computed in $O(|V|)$ total time.  By Proposition \ref{prop:formula}, $$d_G(x,y)\leq \widehat{d}(x,y)\leq 
d_G(x,y)+\delta+1.$$

The values $\widehat{d}(x,y)$ for all $y$ with $y>x$ can be computed using the following simple procedure. We will omit the subscripts $G$ and $T$ if no ambiguities arise. Let also $L_i= \{v\in V: d(v,s)=i\}$. In the procedure, $S_u$ represents vertices of a subtree of $T$ rooted at $u$.

\medskip

\noindent
(01)\hspace*{1cm} set $q:=h(x)$ \\
(02)\hspace*{1cm} define a set $S_u:=\{u\}$ for each vertex $u\in L_q$, $u>x$,  and denote this family of sets by $\cal F$  \\
(03)\hspace*{1cm} {\bf for} $k=q$ downto 0 \\
(04)\hspace*{1.5cm} let $x_k$ be the vertex from $L_k\cap P_T(x,s)$ \\
(05)\hspace*{1.5cm} {\bf for} each vertex $u\in L_k$ with $u>x$ \\
(06)\hspace*{2cm} {\bf if} $d_G(u,x_k)\leq \delta$ (i.e., $u=x_k$ or $u$ is adjacent to $x_k$ in $G^{\delta}$) {\bf then}\\
(07)\hspace*{2.5cm} {\bf for} every $v\in S_u$ \\
(08)\hspace*{3cm}   set $\widehat{d}(x,v):=h(x)+h(v)-2k+\delta$ and  remove $S_u$ from $\cal F$ \\
(09)\hspace*{2.5cm} {\bf endfor} \\
(10)\hspace*{1.5cm} {\bf endfor} \\
(11)\hspace*{1.5cm} /*  update $\cal F$ for the next iteration  */ \\
(12)\hspace*{1.5cm} {\bf if} $k>0$ {\bf then}  \\
(13)\hspace*{2cm} {\bf for} each vertex $u\in L_{k-1}$\\
(14)\hspace*{2.5cm} combine all sets $S_{u_1},\dots,S_{u_\ell}$  from $\cal F$  ($\ell\ge 0$), such that $p_T(u_1)=\ldots=p_T(u_\ell)=u$, \\
(15)\hspace*{2.5cm} into one new set $S_u:=\{u\}\cup S_{u_1}\cup\ldots\cup S_{u_\ell}$   ~~~~~~ /*  when $\ell= 0$, $S_u:=\{u\}$  */
(16)\hspace*{2cm} {\bf endfor} \\
(17)\hspace*{1cm} {\bf endfor} \\
(18)\hspace*{1cm} set also $\widehat{d}(x,s):=h(x)$.

\medskip

Thus, we have the following result.

\begin{theorem} \label{th:dist-power}
Let $G=(V,E)$ be a graph with $\delta$-thin triangles. Given $G^{\delta}$, all distances in $G$ with an additive one-sided error of at
most $\delta+1$ can be found in $O(|V|^2)$ time.
\end{theorem}

To avoid the requirement that $G^{\delta}$ is given in advance, we can use any known fast constant-factor approximation algorithm that
in total $T(|V|)$-time computes for every pair of vertices $x,y$ of  $G$ a value $\widetilde{d}(x,y)$ such that $d_G(x,y)\le \widetilde{d}(x,y)\le \alpha d_G(x,y)+\beta$. We can show that, using such an algorithm as a preprocessing step, the distances in a graph $G$ with $\delta$-thin triangles can be additively approximated with an additive one-sided error of at most $\alpha\delta+\beta+1$ in $O(T(|V|)+|V|^2)$ time.

Although one can use any known fast constant-factor approximation algorithm in the preprocessing step, in what follows, we will demonstrate our idea using a fast approximation algorithm from \cite{BeKa}. It computes in $O(|V|^2\log^2|V|)$ total time for every pair $x,y$ a value  $\widetilde{d}(x,y)$ such that $$d_G(x,y)\le \widetilde{d}(x,y)\le 2 d_G(x,y)+1.$$

Assume that the values $\widetilde{d}(x,y)$, $x,y\in V$,  are  precomputed. By $\widetilde{sl}_T(x,y;\lambda)$ we denote now the largest index $k$ such that $\widetilde{d}_G(x_k,y_k)\leq \lambda$. We have

\begin{proposition} \label{prop:formula2}
For every vertices $x$ and $y$ of a graph $G$ with $\delta$-thin triangles, any integer $\rho \ge \delta$, and any $BFS$-tree $T$ of $G$, $$h_T(x)+h_T(y)-2k-1\leq d_G(x,y)\leq h_T(x)+h_T(y)-2k+d_G(x_k,y_k),$$ where $k=\widetilde{sl}_T(x,y;2\rho+1)$.
\end{proposition}

\begin{proof} The proof is identical to the proof of Proposition \ref{prop:formula2}. One needs only to notice the following. In a geodesic triangle $\Delta(x,y,s):=[x,y]\cup [x,s]\cup [y,s]$ with $[x,s]=P_T(x,s)=(x_q,x_{q-1},\dots,x_1,s)$ and $[y,s]=P_T(y,s)=(y_p,y_{p-1},\dots,y_1,s)$, for each $i\leq (x|y)_s$, $d_G(x_i,y_i)\leq \delta\le \rho$ and, hence, $\widetilde{d}(x_i,y_i)\leq 2\rho+1$ holds. Therefore,  $\widetilde{sl}_T(x,y;2\rho+1)\ge (x|y)_s-\frac{1}{2}$. \qed
\end{proof}

Let $\rho$ be any integer greater than or equal to $\delta$. By replacing in our earlier procedure lines (06) and (08) with
\medskip

\noindent
(06)$'$\hspace*{1cm} {\bf if} $\widetilde{d}(u,x_k)\leq 2\rho+1$ {\bf then}\\
(08)$'$\hspace*{1cm}   set $\widehat{d}(x,v):=h(x)+h(v)-2k+2\rho+1$ and  remove $S_u$ from $\cal F$ \\
\medskip

\noindent
we will compute for each current vertex $x$ all values $\widehat{d}(x,y):=h_T(x)+h_T(y)-2\widetilde{sl}_T(x,y;2\rho+1)+2\rho+1$, $y>x$,
in $O(|V|)$ total time.  By Proposition \ref{prop:formula2},
\begin{eqnarray}
\nonumber  d_G(x,y) &\leq &  h_T(x)+h_T(y)-2\widetilde{sl}_T(x,y;2\rho+1)+d_G(x_k,y_k) \\
\nonumber   &\leq & h_T(x)+h_T(y)-2\widetilde{sl}_T(x,y;2\rho+1)+\widetilde{d}(x_k,y_k) \\
\nonumber  &\leq & h_T(x)+h_T(y)-2\widetilde{sl}_T(x,y;2\rho+1)+2\rho+1 \\
\nonumber   &=& \widehat{d}(x,y)
\end{eqnarray}
and
\begin{eqnarray}
\nonumber  \widehat{d}(x,y) &= &  h_T(x)+h_T(y)-2\widetilde{sl}_T(x,y;2\rho+1)+2\rho+1 \\
\nonumber   &\leq & d_G(x,y)+2\rho+2.
\end{eqnarray}

\medskip

Thus, we have the following result:

\begin{theorem}  \label{th:dist}
Let $G=(V,E)$ be a graph with $\delta$-thin triangles.
 \begin{enumerate}
   \item[$(a)$] If the value of $\delta$ is known, then all distances in $G$ with an additive one-sided error of at
most $2\delta+2$ can be found in $O(|V|^2\log^2|V|)$ time.
   \item[$(b)$] If an approximation $\rho$ of $\delta$ such that $\delta\le \rho\le a\delta+b$ is known (where $a$ and $b$ are constants), then all distances in $G$ with an additive one-sided error of at most $2(a\delta+b+1)$ can be found in $O(|V|^2\log^2|V|)$ time.
 \end{enumerate}
\end{theorem}

The second part of Theorem \ref{th:dist} says that if an approximation of the thinness parameter of a graph $G$ is given then all distances in $G$ can be additively approximated in $O(|V|^2\log^2|V|)$ time.
Recently, it was shown in \cite{fast-appr-hyp} that the following converse is true. From an estimate of all distances in $G$ with an additive one-sided error of at most $k$, it is possible to compute in $O(|V|^2)$ time an estimation $\rho^*$ of the thinness of $G$ such that $\delta \le \rho^*\le 8\delta+12k+4,$ proving a $\tilde{O}(|V|^2)$-equivalence between approximating the thinness and approximating the distances in graphs.

\section{Experimentation on Some Real-World Networks} \label{sec:experim}
In this section, we analyze the performance of our algorithms for approximating eccentricities and distances on a number of real-world networks. Our experimental results show that the estimates on eccentricities and distances obtained are even better than the theoretical bounds described in Corollary \ref{th:rad-diam-appr} and Theorems  \ref{th:all-ecc--appr},\ref{th:dist}.

\begin{table}[htb]
\begin{center}{\fontsize{7}{9} \selectfont
\begin{tabular}{|lcr|rrrr|rrrr|r|}
\hline
\\[-0.25cm]
Network & Type & Ref. & $|V|$ & $|E|$ & $|C(G)|$ & $\overline{deg}$ & $rad(G)$ & $diam(G)$ & $diam_G(C(G))$ & connected? & $\delta(G)$\\ [0.1cm] \hline


{\sc dutch-elite} & \multirow{6}{*}{social} &\cite{Pajek} &  3621 & 4310 & 3 & 2.4 & 12 & 22 & 4 & no & 5\\
{\sc facebook} &  & \cite{leskovec2012learning} & 4039 & 88234 & 1 & 43.7 & 4 & 8 & 0 & yes & 1.5   \\
{\sc eva} &  & \cite{Pajek} & 4475 & 4664 & 15 & 2.1 & 10 & 18 & 3 & yes & 3.5\\
{\sc slashdot}  &  & \cite{leskovec2009community} & 77360 & 905468 & 1 & 13.1 & 6 & 12 & 0 & yes & *1.5\\
{\sc loans}  &  & \cite{loans} & 89171 & 3394979 & 29350 & 74.69 & 5 & 8 & 4 & yes & \\
{\sc twitter} &  & \cite{twitter} & 465017 & 834797 & 755 & 3.59 & 5 & 8 & 4 & yes & \\

\hline

{\sc email-virgili} &  & \cite{guimera2003self} & 1133 & 5451 & 215 & 9.6  & 5 & 8  & 4 & yes & 2  \\
{\sc email-enron}  & \multirow{1}{*}{communi-} & \cite{leskovec2009community,klimmt2004introducing} & 33696 & 180811 & 248 & 10.7 & 7 & 13 & 2 & yes &\\
{\sc email-eu} & \multirow{1}{*}{cation}  & \cite{leskovec2007graph}  & 224832 & 680720 & 1 & $\approx 3$ & 7 & 14 & 0 & yes & \\
{\sc wikitalk-china} &  & \cite{wiki-china} & 1217365 & 3391055 & 17 & 2.9 & 4 & 8 & 2 & yes & \\
\hline

{\sc cs-metabolic} & \multirow{4}{*}{biological} & \cite{CEMeta} & 453 & 4596 & 17 & 8.9 & 4 & 7 & 2 & yes & 1.5 \\
{\sc sc-ppi} &  & \cite{ppi} & 1458 & 1948 & 48 & 2.7 & 11 & 19  & 6 & no & 3.5\\
{\sc yeast-ppi} &  & \cite{yeast} & 2224 & 6609 & 57 & $\approx 6$ & 6 & 11 & 4 & no & 2.5 \\
{\sc homo-pi} &  & \cite{biogrid} & 16635 & 115364 & 135 & 13.87 & 5 & 10 & 2 & no & 2 \\ \hline

{\sc as-graph-1} & \multirow{6}{*}{internet} & \cite{ASGraphs} & 3015 & 5156 & 32 & 3.4 & 5 & 9 & 2 & yes & 2 \\
{\sc as-graph-2} &  & \cite{ASGraphs} & 4885 & 9276 & 531 & 3.8 & 6 & 11 & 4 & no & 3 \\
{\sc as-graph-3} &  & \cite{ASGraphs} & 5357 & 10328 &  10 & 3.9 & 5 & 9 & 2 & yes & 2\\
{\sc routeview}  &  & \cite{oregon} & 10515 & 21455 & 2 & 4.1 & 5 & 10 & 2 & no & 2.5 \\
{\sc as-caida}  &  & \cite{caida} & 26475 & 53381 & 2 & 4.03 & 9 & 17 &  1 & yes & 2.5 \\
{\sc itdk}  &  & \cite{itdk} & 190914 & 607610 & 155 & 6.4 & 14 & 26 & 4 & yes & \\ \hline

{\sc gnutella-06} & \multirow{4}{*}{peer-to-peer} & \cite{ripeanu2002mapping,leskovec2007graph} & 8717  & 31525 & 338 & 7.2 & 6 & 10 & 5 & no & 3\\
{\sc gnutella-24} &  & \cite{ripeanu2002mapping,leskovec2007graph} & 26498 & 65359 & 1 & 4.9 & 6 & 11 & 0 & yes & 3\\
{\sc gnutella-30} &  & \cite{ripeanu2002mapping,leskovec2007graph} &  36646 & 88303 & 602 & 4.8 & 7 & 11 & 6 & no & *2.5\\
{\sc gnutella-31} &  & \cite{ripeanu2002mapping,leskovec2007graph} & 62561 & 295756 & 55 & 4.7 & 7 & 11 & 5 & no & *2.5\\ \hline

{\sc web-stanford} & \multirow{3}{*}{web} & \cite{leskovec2009community}  &  255265 & 2234572 & 1 & 15.2 & 82 & 164 & 0 & yes & *7\\
{\sc web-notredam} &  & \cite{albert1999internet} & 325729 & 1497134 & 12 & 6.8 & 23 & 46 & 2 & no & *2\\
{\sc web-berkstan} &  & \cite{leskovec2009community} & 654782 & 7600595 & 1 & 20.1 & 104 & 208 & 0 & yes & *7\\ \hline

{\sc amazon-1} & \multirow{1}{*}{product }  & \cite{yang2012defining} &  334863 & 925872 & 21 & 5.5 & 24 & 47 & 3 & no &\\
{\sc amazon-2} & \multirow{1}{*}{co-purchasing}  & \cite{yang2012defining} &  400727 & 3200440 & 194 & 11.7 & 11 & 20 & 5 & no &\\ \hline

{\sc road-euro} & \multirow{4}{*}{infrastructure} & \cite{euroroad} & 1039 & 1305 & 1 & 2.5 & 31 & 62 & 0 & yes & 7.5 \\
{\sc openflight} &  & \cite{openflight} & 3397 & 19231 & 21 & 11.3 & 7 & 13 & 2 & yes & 2\\
{\sc power-grid} &  & \cite{powerGrid} & 4941 & 6594 & 1 & 2.7 & 23 & 46 & 0 & yes & 10\\
{\sc road-pa} &  & \cite{leskovec2009community} & 1087562 & 3083028 & 2 & 2.83 & 402 & 794 & 1 & yes & *195.5\\

\hline
\end{tabular}
 \par }\medskip
\caption{Statistics of the analyzed networks: $|V|$ is the number of vertices, $|E|$ is the number of edges; $|C(G)|$ is the number of central vertices; $\overline{deg}$ is the  average degree; $rad(G)$ is the graph's radius; $diam(G)$ is the graph's diameter; $diam_G(C(G))$ is the  diameter of the graph's center; "connected?" indicates whether or not the center of the graph is connected; $\delta(G)$ is the graph's hyperbolicity. Hyperbolicity values marked with asterisks are approximate.}
\label{table:dataset}
\end{center} \vspace*{-1.0cm}
\end{table}

We apply our algorithms to six social networks,
four email communication networks,
four  biological networks, six internet graphs, four peer-to-peer networks, three web networks, two product-co-purchasing networks, and four infrastructure networks. Most of the networks listed are part of the Stanford Large Network Dataset Collection ({\sc snap}) and the Koblenz Network Collection ({\sc konect}), and are available at \cite{snap} and \cite{konect}. Characteristics of these networks, such as the number of vertices and edges, the average degree, the radius and the diameter, are given in Table \ref{table:dataset}. The numbers listed in Table \ref{table:dataset} are based on the largest connected component of each network, when the entire network is disconnected. We ignore the directions of the edges and remove all self-loops from each network. Additionally, in Table \ref{table:dataset}, for each network we report the size (as the number of vertices) of its center $C(G)$. We also analyze the diameter and the connectivity of the center of each network. The diameter of the center $diam_G(C(G))$ is defined as the maximum distance between any two central vertices in the graph. In the last column of Table \ref{table:dataset}, we report the Gromov hyperbolicity $\delta$ of majority of networks\footnote{All $\delta$-hyperbolicity values listed in Table \ref{table:dataset} were computed using Gromov's four-point condition definition. As mentioned in \cite{GhHa,Gr}, geodesic triangles of geodesic $\delta$-hyperbolic spaces are 4$\delta$-thin.}. Computing the hyperbolicity of a graph is computationally expensive; therefore, we provide the exact $\delta$ values for the smaller networks (those with $|V| \leq 30$K) in our dataset (in some cases, the algorithm proposed in \cite{cohen2012computing} was used). For some larger networks, the approximated $\delta$-hyperbolicity values listed in Table \ref{table:dataset} are as reported in \cite{KeSN16}\footnote{For {\sc web-stanford} and {\sc web-berkstan}, \cite{KeSN16} gives 1.5 and 2, respectively, as estimates on the hyperbolicities. However, the sampling method they used seems to be not very accurate. According to \cite{hakeem-exp}, the hyperbolicities are at least 7 for both graphs.}.
Most networks that we included in our dataset are hyperbolic. However, for comparison reasons, we included also a few infrastructure networks that are known to lack the hyperbolicity property.

\subsection{Estimation of Eccentricities}

Following Proposition \ref{prop:middlepoint-ecc}, for each graph in our dataset, we found a pair $u,v$ of mutually distant vertices. In column two of Table   \ref{table:Method1-results}, we report on how many $BFS$ sweeps of a graph were needed to locate $u$ and $v$. Interestingly, for almost all graphs (28 out 33) only two sweeps were sufficient. For four other graphs (including {\sc road-pa} network whose hyperbolicity is large)  three sweeps were needed, and only for one graph ({\sc power-grid} network) we needed four sweeps.

\begin{table}[htb]
\begin{center}{\fontsize{7}{9} \selectfont
\begin{tabular}{|l|ccc|cc|cc|cc|}
\hline
Network & \begin{tabular}[x]{@{}c@{}}No. of BFS \\ iterations \end{tabular} & $d_G(u,v)$ & $2rad(G) - d_G(u,v)$ & $ecc_G(c)$ & $ecc_G(c) - rad(G)$ & $d_G(c,C(G))$ &  \begin{tabular}[x]{@{}c@{}}$\min i:$ \\ $B(c,i) \supseteq C(G)$ \end{tabular} & $k_{max}$ & $k_{avg}$ \\

&& \textit{Prop.\ref{prop:middlepoint-ecc}(c)} & \textit{Prop.\ref{prop:middlepoint-ecc}(c)} & \textit{Prop.\ref{prop:middlepoint-ecc}(b)} & \textit{Prop.\ref{prop:middlepoint-ecc}(b)} & \textit{Prop.\ref{prop:middlepoint-ecc}(d)} &&& \\\hline

{\sc dutch-elite} & 2 & 22 & 2 & 13 & 1 & 1 & 3 & 6 & 2.35 \\
{\sc facebook}  & 2 & 8 & 0 & 4 & 0 & 0 & 0 & 2 & 0.686 \\
{\sc eva}  & 2 & 18 & 2 & 10 & 0 & 0 & 2 & 2 & 0.571 \\
{\sc slashdot}   & 2 & 11 & 1 & 7 & 1 & 2 & 2 & 3 & 1.777   \\
{\sc loans}  & 2 & 7 & 3 & 5 & 0 & 0 & 3 & 3 & 2.06 \\
{\sc twitter} & 2 & 8 & 2 & 6 & 1 & 1 & 3 & 4 & 2.569 \\
\hline

{\sc email-virgili}  & 2 & 7 & 3 & 6 & 1 & 1 & 3 & 4 & 2.729 \\
{\sc email-enron}  & 2 & 13 & 1 & 7 & 0 & 0 & 2 & 2 & 0.906 \\
{\sc email-eu} & 2 & 14 & 0 & 7 & 0 & 0 & 0 & 2 & 0.002   \\
{\sc wikitalk-china} & 2 & 7 & 1 & 5 & 1 & 1 & 2 & 3 & 2.076\\
\hline

{\sc ce-metabolic} & 2 & 7 & 1 & 5 & 1 & 1 & 2 & 3 & 1.982 \\
{\sc sc-ppi} & 3 & 19 & 3 & 12 & 1 & 2 & 6 & 3 & 0.981 \\
{\sc yeast-ppi} & 3 & 11 & 1 & 6 & 0 & 0 & 3 & 3 & 1.872 \\
{\sc homo-pi}   & 2 & 10 & 0 & 5 & 0 & 0 & 2 & 2 & 0.747 \\ \hline

{\sc as-graph-1}  & 2 & 8 & 2 & 6 & 1 & 1 & 2 & 3 & 1.791\\
{\sc as-graph-2}   & 3 & 11 & 1 & 6 & 0 & 0 & 3 & 3 & 1.124 \\
{\sc as-graph-3}  & 2 & 9 & 1 & 5 & 0 & 0 & 2 & 2 & 0.828 \\
{\sc routeview}     & 2 & 10 & 0 & 5 & 0 & 0& 2  & 2 & 0.329 \\
{\sc as-caida}      & 2 & 17 & 1 & 9 & 0 & 0 & 1 & 0 & 0   \\
{\sc itdk}   & 2 & 26 & 2 & 15 & 1 & 1 & 3 & 4 & 2.108  \\
\hline

{\sc gnutella-06}  & 2 & 10 & 2 & 6 & 0 & 0 & 4 & 4 & 2.507 \\
{\sc gnutella-24}  & 2 & 10 & 2 & 7 & 1 & 1 & 1 & 5 & 2.697 \\
{\sc gnutella-30}  & 2 & 11 & 3 & 7 & 0 & 0 & 5 & 5 & 3.167 \\
{\sc gnutella-31}  & 2 & 11 & 3 & 8 & 1 & 2 & 5 & 6 & 4.176 \\ \hline

{\sc web-stanford}  & 2 & 164 & 0 & 82 & 0 & 0 & 0 & 28 & 0.006 \\
{\sc web-notredam}  & 2 & 46 & 0 & 23 & 0 & 0 & 2 & 2 & 0.935 \\
{\sc web-berkstan}  & 2 & 208 & 0 & 104 & 0 & 0 & 0 & 22 & 0.002 \\ \hline

{\sc amazon-1}  & 2 & 47 & 1 & 24 & 0 & 0 & 2 & 6 & 0.991 \\
{\sc amazon-2}  & 2 & 20 & 2 & 12 & 1 & 2 & 5 & 6 & 3.735 \\ \hline

{\sc road-euro}   & 2 & 62 & 0 & 31 & 0 & 0 & 0 & 8 & 0.135 \\
{\sc openflight}  & 2 & 13 & 1 & 8 & 1 & 1 & 2 & 3 & 1.879 \\
{\sc power-grid}   & 4 & 46 & 0 & 28 & 5 & 8 & 8 & 13 & 5.735 \\
{\sc road-pa}  & 3 & 794 & 10 & 415 & 13 & 44 & 45 & 98 & 23.339 \\

\hline
\end{tabular}
 \par }
\end{center}\medskip
\caption{Qualities of a pair of mutually distant vertices $u$ and $v$, of a middle vertex $c$ of a $(u,v)$-geodesic, and of a $BFS(c)$-tree $T_1$ rooted at vertex $c$. "No. of BFS iterations`` indicates how many breadth-first-search iterations were needed to obtain a pair of mutually distant vertices $u$ and $v$.
For each vertex $x \in V$, $k(x): = ecc_{T_1}(x) - ecc_G(x)$. Also, $k_{max} := \max_{x\in V} k(x)$ and  $k_{avg} := \frac{1}{n} \sum_{x \in V} k(x)$.}
\label{table:Method1-results}
\vspace*{-1.0cm}
\end{table}

In column four of  Table  \ref{table:Method1-results}, we report for each graph $G$ the difference between  $2rad(G)$ and $d_G(u,v)$. Proposition \ref{prop:middlepoint-ecc}(c) says that the difference must be at most $2\delta+1$, where $\delta$ is the thinness of geodesic triangles in $G$. Actually, for large number (27 out of 33) of graphs in our dataset, the difference is at most two. Five other graphs have the difference equal to 3, and only
{\sc road-pa} network has the difference equal to 10. We have $d_G(u,v)=diam(G)$  for 27 graphs in our dataset, including {\sc road-pa} network whose geodesic triangles thinness is at least 196. For remaining six graphs $d_G(u,v)=diam(G)-1$ holds.

We also analyzed the quality of a middle vertex $c$ of a randomly picked shortest path between mutually distant vertices $u$ and $v$.  Proposition \ref{prop:middlepoint-ecc} states that $ecc_G(c)$ is close to $rad(G)$ and $c$ is not too far from the graph's center $C(G)$. Table \ref{table:Method1-results} lists the properties of the selected middle vertex $c$. In almost all graphs, vertex $c$ belongs to the center $C(G)$ or is at distance one or two from $C(G)$. Even in graphs with $ecc_G(c)- rad(G) > 2$ ({\sc power-grid} and {\sc road-pa}), the value $ecc_G(c)- rad(G)$ is smaller than what is suggested by Proposition \ref{prop:middlepoint-ecc}(b). It is also clear from Table \ref{table:Method1-results} that $c$ is not too far from any vertex in $C(G)$ (look at the radius $i$ of the ball $B(c,i)$ required to include $C(G)$). In all graphs, $i$ is much smaller than $2\delta+1$ (indicated in  Proposition \ref{prop:middlepoint-ecc}(d)).

\begin{table}[htb]
\begin{center}{\fontsize{7}{9} \selectfont
\begin{tabular}{|l|ccc|cc|cc|}
\hline
Network &  $ecc_G(v)$ & $2rad(G) - ecc_G(v)$ & $ecc_G(w)$ & $d_G(w,C(G))$ & \begin{tabular}[x]{@{}c@{}}$\min i:$ \\ $B(w,i) \supseteq C(G)$ \end{tabular} & $k_{max}$ & $k_{avg}$ \\

&& \textit{Prop.\ref{prop:r-vertex-ecc}(c)} & \textit{Prop.\ref{prop:r-vertex-ecc}(a)} &&&& \\\hline

{\sc dutch-elite} &  22 & 2 & 12 & 0 & 4 & 6 & 2.431\\
{\sc facebook} &  8 & 0 & 5 & 3 & 3 & 3 & 0.704 \\
{\sc eva} &  18 & 2 & 11 & 1 & 3 & 2 & 0.572 \\
{\sc slashdot}    &  11 & 1 & 7 & 2 & 2 & 3 & 1.88 \\
{\sc loans}  & 7 & 3 & 5 & 0 & 3 & 3 & 2.031 \\
{\sc twitter}&  8 & 2 & 5 & 0 & 3 & 3 & 1.821 \\
\hline

{\sc email-virgili} &  7 & 3 & 5 & 0 & 4 & 4 & 1.932 \\
{\sc email-enron} &  13 & 1 & 7 & 0 & 2 & 2 & 0.903 \\
{\sc email-eu} &  14 & 0 & 7 & 0 & 0 & 2 & 0.002 \\
{\sc wikitalk-china} &  8 & 0 & 5 & 1 & 2 & 3 & 1.791 \\
\hline

{\sc ce-metabolic} &  7 & 1 & 4 & 0 & 1 & 1 & 0.349 \\
{\sc sc-ppi} &  19 & 3 & 12 & 1 & 6 & 7 & 4.196 \\
{\sc yeast-ppi} &  11 & 1 & 7 & 1 & 3 & 4 & 2.558 \\
{\sc homo-pi}        &  9 & 1 & 5 & 0 & 2 & 2 & 0.612 \\\hline

{\sc as-graph-1} &  9 & 1 & 5 & 0 & 2 & 2 & 0.887 \\
{\sc as-graph-2} &  11 & 1 & 6 & 0 & 3 & 2 & 0.833\\
{\sc as-graph-3} &  9 & 1 & 5 & 0 & 2 & 2 & 0.312 \\
{\sc routeview}    &  10 & 0 & 5 & 0 & 2 & 2 & 0.329 \\
{\sc as-caida}     &  17 & 1 & 9 & 0 & 1 & 0 & 0  \\
{\sc itdk}        &  26 & 2 & 15 & 1 & 3 & 5 & 2.702 \\ \hline

{\sc gnutella-06} &  10 & 2 & 7 & 1 & 5 & 5 & 3.543\\
{\sc gnutella-24} &  11 & 1 & 8 & 3 & 3 & 6 & 4.475\\
{\sc gnutella-30} &  11 & 3 & 8 & 1 & 5 & 6 & 4.034\\
{\sc gnutella-31} &  11 & 3 & 8 & 1 & 5 & 6 & 4.251\\ \hline

{\sc web-stanford}&   164 & 0 & 82 & 0 & 0 & 28 & 0.006 \\
{\sc web-notredam} &   46 & 0 & 23 & 0 & 2 & 2 & 0.935 \\
{\sc web-berkstan} &   208 & 0 & 104 & 0 & 0 & 22 & 0.002 \\ \hline

{\sc amazon-1} &   47 & 1 & 24 & 0 & 3 & 7 & 0.919\\
{\sc amazon-2} &    20 & 2 & 11 & 0 & 5 & 5 & 2.03\\ \hline

{\sc road-euro}  &   62 & 0 & 31 & 0 & 0 & 8 & 0.135\\
{\sc openflight} &   13 & 1 & 7 & 0 & 2 & 2 & 0.641 \\
{\sc power-grid}  &   46 & 0 & 23 & 0 & 0 & 4 & 1.409 \\
{\sc road-pa} &   772 & 32 & 417 & 21 & 22 & 80 & 22.545 \\
\hline
\end{tabular}
\par} \medskip
\caption{Qualities of a vertex $v$ most distant from a random vertex $u$, of a vertex $w$ of a $(u,v)$-geodesic at distance $rad(G)$ from $v$, and of a $BFS(w)$-tree $T_2$ rooted at vertex $w$. For each vertex $x \in V$, $k(x): = ecc_{T_2}(x) - ecc_G(x)$. Also, $k_{max} := \max_{x\in V} k(x)$ and  $k_{avg} := \frac{1}{n} \sum_{x \in V} k(x)$.}
\label{table:Method2-results}
\end{center}
\vspace*{-1.0cm}
\end{table}

Following Theorem  \ref{th:ecc-appr-tree}, for each graph $G=(V,E)$ in our dataset, we constructed an arbitrary $BFS(c)$-tree $T_1=(V,E')$, rooted at vertex $c$, and analyzed how well $T_1$ preserves or approximates the eccentricities of vertices in $G$.
By Theorem  \ref{th:ecc-appr-tree},  $ecc_G(v)\leq ecc_{T_1}(v)\leq ecc_G(v)+ 3\delta+1$ holds for every $v\in V$.
In our experiments, for each graph $G$ and the constructed  for it $BFS(c)$-tree $T_1$, we computed $k_{max} := \max_{v\in V} \{ecc_{T_1}(v) - ecc_G(v)\}$ ({\em maximum distortion}) and $k_{avg} := \frac{1}{n}\sum_{v \in V} ecc_{T_1}(v) - ecc_G(v)$ ({\em average distortion}).
For most graphs (see Table \ref{table:Method1-results}), the value of $k_{max}$ is small: $k_{max} = 0$ for one graph, $k_{max} = 2$ for eight graphs, $k_{max} = 3$ for nine graphs, $k_{max} = 4$ for four graphs, $k_{max} = 5$ for two graphs, and $k_{max} > 5$ for nine graphs. Also, the average distortion $k_{avg}$ is much smaller than $k_{max}$ for all graphs. In fact, $k_{avg} < 3$ in all but five graphs ({\sc gnutella-30, gnutella-31, amazon-2, power-grid,} and {\sc road-pa}). In graphs with high $k_{max}$, close inspection reveals that only small percent of vertices achieve this maximum. For example, in graph  {\sc web-stanford}, $k_{max} = 28$ was only achieved by 17 vertices.  The distributions of the values of $k(v) := ecc_{T_1}(v) - ecc_G(v)$ of all graphs are listed in Table \ref{table:k-distributions-method1} (see Appendix).

Similar experiments were performed following Proposition \ref{prop:r-vertex-ecc}. 
For each graph $G$ in our dataset,  we picked a random vertex $u\in V$ and a random vertex $v\in F(u)$. Then, we identified in a randomly picked $(u,v)$-geodesic a vertex $w$ at distance $rad(G)$ from $v$. We did not consider a vertex $c$ defined in Proposition \ref{prop:r-vertex-ecc}(d) since,
for majority of graphs in our dataset,  $c$ will be a middle vertex of a geodesic between two mutually distant vertices, and working with $c$ we will duplicate previous experiments. Recall that for majority of our graphs (as found in our experiments) two BFS sweeps already identify a pair of mutually distant vertices.
We know  from Proposition \ref{prop:r-vertex-ecc} that $ecc_G(v)\ge diam(G)-2\delta \ge 2rad(G)-4\delta-1$ and $ecc_G(w)\leq rad(G)+\delta$.  Our experimental results are better than these theoretical bounds. In Table \ref{table:Method2-results}, we list eccentricities of $v$ and $w$ for each graph. In almost all graphs, the eccentricity of $v$ is equal to the diameter $diam(G)$. Only four graphs have $ecc_G(v)=diam(G)-1$ and one graph ({\sc road-pa}) has $ecc_G(v)>diam(G)-1$.  Vertex $w$ is central for 21 graphs, has eccentricity equal to $rad(G)+1$ for 10 graphs, has eccentricity equal to $rad(G)+2$ for one graph, and only for one remaining graph ({\sc road-pa} network, which has large hyperbolicity) its eccentricity is equal to $rad(G)+15$.
It turns out also (see columns five and six of Table \ref{table:Method1-results}) that vertex $w$ either belongs to the center $C(G)$ or is very close to the center. The only exception is again {\sc road-pa} network  where $2rad(G) - ecc_G(w) = 32$ and $d(w,C(G)) = 21$.

For every graph $G=(V,E)$ in our dataset, we constructed also an arbitrary $BFS(w)$-tree $T_2=(V,E')$, rooted at vertex $w$, and analyzed how well $T_2$ preserves or approximates the eccentricities of vertices in $G$.
The value of $k_{max}$ is at most five for 23 graphs. The average distortion $k_{avg}$ is much smaller than $k_{max}$ in all graphs. The distributions of the values of $k(x)$ for all graphs are presented in Table \ref{table:k-distributions-method2} (see Appendix).

\begin{table}[htb]
\begin{center}{\fontsize{7}{9} \selectfont
\begin{tabular}{|l|c|ccc|ccc|ccc|}
\hline
 &   &   &   &   &   &   &  &   &   &   \\
Network & $diam(G)$ & $diam(T_1)$ & $k_{max}^{T_1}$ & $k_{avg}^{T_1}$ & $diam(T_2)$ & $k_{max}^{T_2}$ & $k_{avg}^{T_2}$ & $diam(T_3)$ & $k_{max}^{T_3}$ & $k_{avg}^{T_3}$ \\

\hline

{\sc dutch-elite} & 22 & 24 & 6 & 2.35 & 24 & 6 & 2.431 & 24 & 6 & 2.083\\
{\sc facebook} & 8 & 8 & 2 & 0.686 & 9 & 3 & 0.704 & 8 & 2 & 0.686 \\
{\sc eva} & 18 & 19 & 2 & 0.571 & 19 & 2 & 0.572 & 19 & 2 & 0.571 \\
{\sc slashdot}    & 12 & 14 & 3 & 1.777 & 14 & 3 & 1.88 & 12 & 2 & 0.701 \\
{\sc loans}  & 8 & 10 & 3 & 2.06 & 10 & 3 & 2.031 & 10 & 3 & 2.081 \\
{\sc twitter} & 8 & 11 & 4 & 2.569 & 10 & 3 & 1.821 & 10 & 4 & 1.856 \\
\hline
{\sc email-virgili} & 8 & 11 & 4 & 2.729 & 10 & 4 & 1.932 & 10 & 4 & 1.906 \\
{\sc email-enron} & 13 & 13 & 2 & 0.906 & 14 & 2 & 0.903 & 14 & 2 & 1.735 \\
{\sc email-eu} & 14 & 14 & 2 & 0.002 & 14 & 2 & 0.002 & 14 & 2 & 0.002 \\
{\sc wikitalk-china} & 8 & 9 & 3 & 2.076 & 9 & 3 & 1.791 & 8 & 2 & 0.777\\
\hline

{\sc ce-metabolic} & 7 & 9 & 3 & 1.982 & 8 & 1 & 0.349 & 8 & 2 & 1.185\\
{\sc sc-ppi} & 19 & 20 & 3 & 0.981 & 23 & 7 & 4.196 & 22 & 6 & 3.163\\
{\sc yeast-ppi} & 11 & 12 & 3 & 1.872 & 13 & 4 & 2.558 & 12 & 3 & 1.872 \\
{\sc homo-pi}     & 10 & 10 & 2 & 0.747 & 10 & 2 & 0.612 & 10 & 2 & 0.747\\
\hline

{\sc as-graph-1} & 9 & 11 & 3 & 1.791 & 10 & 2 & 0.887 & 10 & 2 & 0.886 \\
{\sc as-graph-2} & 11 & 11 & 3 & 1.124 & 11 & 2 & 0.833 & 12 & 3 & 1.272\\
{\sc as-graph-3} & 9 & 10 & 2 & 0.828 & 10 & 2 & 0.312 & 10 & 2 & 0.312\\
{\sc routeview}    & 10 & 10 & 2 & 0.329 & 10 & 2 & 0.329 & 10 & 2 & 0.329\\
{\sc as-caida}     & 17 & 17 & 0 & 0 & 17 & 0 & 0 & 17 & 0 & 0 \\
{\sc itdk}        & 26 & 29 & 4 & 2.108 & 29 & 5 & 2.702 & 28 & 3 & 1.385 \\
\hline

{\sc gnutella-06} & 10 & 12 & 4 & 2.507 &13 & 5 & 3.543 & 12 & 4 & 2.507\\
{\sc gnutella-24} & 11 & 14 & 5 & 2.697 & 16 & 6 & 4.475 & 12 & 3 & 0.863\\
{\sc gnutella-30} & 11 & 14 & 5 & 3.167 & 16 & 6 & 4.034 & 14 & 5 & 3.295\\
{\sc gnutella-31} & 11 & 16 & 6 & 4.176 & 16 & 6 & 4.251 & 14 & 5 & 2.669 \\
\hline

{\sc web-stanford}& 164 & 164 & 28 & 0.006 & 164 & 28 & 0.006 & 164 & 28 & 0.006 \\
{\sc web-notredam}& 46 & 46 & 2 & 0.935 & 46 & 2 & 0.935 & 46 & 2 & 0.017\\
{\sc web-berkstan}& 208 & 208 & 22 & 0.002 & 208 & 22 & 0.002 & 208 & 22 & 0.002 \\
\hline

{\sc amazon-1}    & 47 & 47 & 6 & 0.991 & 48 & 7 & 0.919& 47 & 7 & 1.205 \\
{\sc amazon-2}    & 20 & 23 & 6 & 3.735 & 22 & 5 & 2.03 & 22 & 4 & 1.274 \\
\hline

{\sc road-euro}  & 62  & 62 & 8 & 0.135 & 62 & 8 & 0.135 & 62 & 8 & 0.135\\
{\sc openflight} & 13  & 15 & 3 & 1.879 & 14 & 2 & 0.641 & 14 & 2 & 0.704\\
{\sc power-grid} & 46  & 51 & 13 & 5.735 & 46 & 4 & 1.409 & 46 & 4 & 1.409\\
{\sc road-pa}    & 794 & 814 & 98 & 23.339 & 830 & 80 & 22.545  & 803 & 46 & 10.64\\
\hline
\end{tabular}
\par }\medskip
\caption{Comparison of three BFS-trees $T_1$, $T_2$ and $T_3$.  $T_3$ is a $BFS(c^*)$-tree rooted at a randomly picked central vertex $c^*\in C(G)$.}
\label{table:comparison}
\end{center}
\vspace*{-1.0cm}
\end{table}

In Table \ref{table:comparison}, we compare these two eccentricity approximating spanning trees $T_1$ and $T_2$ with each other and with a third $BFS(c^*)$-tree $T_3$ which we have constructed starting from a randomly chosen central vertex $c^*\in C(G)$.

For each graph in the dataset, three values of $k_{max}$ ($k_{max}^{T_1}$, $k_{max}^{T_2}$ and $k_{max}^{T_3}$) and three values of $k_{avg}$ ($k_{avg}^{T_1}$, $k_{avg}^{T_2}$ and $k_{avg}^{T_3}$) are listed. We observe that the smallest $k_{max}$ (out of three) is achieved by tree $T_3$ in 28 graphs, by tree $T_2$ in 20 graphs and by tree $T_1$ in 20 graphs (in 14 graphs, the smallest $k_{max}$ is achieved by all three trees). The difference between the largest and the smallest $k_{max}$ of a graph is at most one for 26 graphs in the dataset. The largest difference is observed for {\sc road-pa} network:  the largest $k_{max}$ (98) is given by tree $T_1$, the smallest $k_{max}$ (46) is given by tree $T_3$. Two other graphs have the difference larger than three: for {\sc sc-ppi} network,  the largest $k_{max}$ (7) is given by tree $T_2$, the smallest $k_{max}$ (3) is given by tree $T_1$; for {\sc power-grid} network,  the largest $k_{max}$ (13) is given by tree $T_1$, the smallest $k_{max}$ (4) is shared by remaining trees $T_2$, $T_3$. Overall, we conclude that $k_{max}$ values for trees $T_1$ and $T_2$ are comparable and generally can be slightly worse than those for tree $T_3$. Similar observations hold also for the average distortion $k_{avg}$. Note, however, that for construction of trees $T_2$ and $T_3$ one needs to know $rad(G)$ or a central vertex of $G$, which are unlikely to be computable in subquadratic time (see Statement \ref{stat:lower-rad}).


\begin{table}[htb]
\begin{center}{\fontsize{7}{9} \selectfont
\begin{tabular}{|l|cc|ccc|c|}
\hline
Network & $diam(G)$ & $rad(G)$ & ${\delta}$ & $\Delta_{max}({\delta})$ & $\Delta_{avg}({\delta})$ &    $ecc(s)$

\\ \hline

{\sc dutch-elite} &  22 &  12 & 8 & 8 & 0.177 &   16  \\
{\sc facebook} & 8 &  4 & 2 & 2 & 0.169 &   6  \\
{\sc eva} &  18 &  10 & 6 & 6 & 0.044 &   12 \\
{\sc slashdot*}    &  12 & 6 & 4 & 2 & 0.028 &   8 \\
{\sc loans*}  &  8 & 5 & 3 & 3 & 0.213 &   6\\
{\sc twitter*}  & 8 & 5 & 3 & 3 & 0.156 &    6\\
\hline

{\sc email-virgili} & 8 & 5  & 3 & 4 & 0.39 &    6 \\
{\sc email-enron} &  13 & 7  & 4 & 4 & 0.06 &   9  \\
{\sc email-eu*} & 14 & 7 & 3 & 2 &  0.005 &    10  \\

\hline

{\sc ce-metabolic} & 7 &  4 & 2 & 3 & 0.125 &   4 \\
{\sc sc-ppi} &  19 & 11  & 6 & 6 & 0.19 &   13 \\
{\sc yeast-ppi} & 11 & 6  & 4 & 4 & 0.239 &   8 \\
{\sc homo-pi}   & 10 & 5  & 3 & 3 & 0.02 &   7   \\ \hline

{\sc as-graph-1} & 9 & 5 & 3 & 4 & 0.061 &   8 \\
{\sc as-graph-2} & 11 & 6  & 4 & 4 & 0.034 &   8  \\
{\sc as-graph-3} &  9 & 5 & 4 & 3 & 0.035 &   9  \\
{\sc routeview}    &  10 & 5  & 4 & 4 & 0.038 &   6  \\
{\sc as-caida}     &  17 & 9 & 3 & 4 & 0.022 &    14  \\
{\sc itdk*}  & 26 & 14 & 5 & 4 & 0.15 &   19 \\
\hline

{\sc gnutella-06} & 9 & 6  & 5 & 4 & 0.331 &   8 \\
{\sc gnutella-24} & 11 & 6 & 6 & 6 &  0.128 &   9  \\
{\sc gnutella-30*} &  11 & 7 & 6 & 5 & 0.439 &   8 \\
{\sc gnutella-31*} & 11 & 7 & 6 & 5 & 0.386 &    9 \\ \hline

\hline

{\sc road-euro}  &  62 & 31  & 21 & 11 & 0.927 &   39\\
{\sc openflight} &  13 & 7  &  3 & 4 & 0.029 &   10 \\
{\sc power-grid} &  46 & 23 &  17 & 17 & 0.518 &   38 \\

\hline
\end{tabular}
\par}\medskip
\caption{Distance approximations: for every $x,y\in V$, $\Delta_{xy}(\delta) =  \widehat{d}_{\delta}(x,y) - d_G(x,y)$; $\Delta_{max}(\delta) = \max_{x,y\in V} \Delta_{xy}(\delta)$; $\Delta_{avg}(\delta) = \frac{1}{n^2}\sum_{x,y\in V} \Delta_{xy}(\delta) $; ${\delta}$ is defined as the smallest $\delta$ ($0 \leq  \delta \leq diam(G)$) such that
$\Delta_{max}(\delta)\le \delta +1$. Due to large sizes of some networks, the values of $\Delta_{max}({\delta})$ and $\Delta_{avg}({\delta})$ for networks marked with * were computed only for some sampled vertices (we sampled vertices that are most distant from the root). The number of sampled vertices ranged from 10 to 100 in each network.}
\label{table:distance_approximation}
\end{center}
\end{table}

\subsection{Estimation of Distances} Following Theorem \ref{th:dist-power}, we experimented also on how well our approach approximates the distances in  graphs from our dataset.
 To analyze the quality of approximation provided by our method for a given graph $G=(V,E)$, for every $\delta:= 0,1,2,\dots$, we computed an estimate $\widehat{d_{\delta}}(x,y)$ on $d_G(x,y)$ and the error $\Delta_{xy}(\delta) =  \widehat{d_{\delta}}(x,y) - d_G(x,y)$ for all $x,y\in V$. In Table \ref{table:distance_approximation}, we report
$\Delta_{max}(\delta) = \max_{x,y\in V} \Delta_{xy}(\delta)$ and $\Delta_{avg}(\delta) = \frac{1}{n^2}\sum_{x,y\in V} \Delta_{xy}(\delta)$ for the smallest $\delta$ such that $\Delta_{max}(\delta)\le \delta +1$. We omitted some very large graphs in this experiment. For some other large graphs, we did only sampling; we calculated
$\Delta_{max}({\delta})$ and $\Delta_{avg}({\delta})$ based only on a set of sampled vertices. We sampled vertices that are most distant from the root. The number of sampled vertices ranged from 10 to 100 in each network. For all networks investigated, the average error $\Delta_{avg}(\delta)$ was very small, less that 1 even for infrastructure networks. That is, the  maximum error $\Delta_{max}(\delta)$ was realized on a very small number of vertex pairs. The  maximum error $\Delta_{max}(\delta)$ was 2 for three networks, was 3 for five networks,  was 4 for ten networks (including  infrastructure network {\sc openflight}), and was at most 6 for all except one social network {\sc dutch-elite} and two infrastructure networks: {\sc road-euro} and {\sc power-grid}. The largest $\Delta_{max}(\delta)$ value had expectedly {\sc power-grid} network whose hyperbolicity is 10.

\subsection*{Acknowledgements}
 The research of V.C., M.H., and Y.V.  was
 supported by ANR project DISTANCIA (ANR-17-CE40-0015).

\newpage
{\large Appendix}

\begin{table}[htb]
\begin{center}{\fontsize{7}{9} \selectfont
\begin{tabular}{|l|cc|ccccccc|}
\hline
Network & $k_{max}$ & $k_{avg}$ & \begin{tabular}[x]{@{}c@{}}\% of vertices\\with $k(x) = 0$\end{tabular} & \begin{tabular}[x]{@{}c@{}}\% of vertices\\with $k(x) = 1$\end{tabular} & \begin{tabular}[x]{@{}c@{}}\% of vertices\\with $k(x) = 2$\end{tabular} & \begin{tabular}[x]{@{}c@{}}\% of vertices\\with $k(x) = 3$\end{tabular} & \begin{tabular}[x]{@{}c@{}}\% of vertices\\with $k(x) = 4$\end{tabular} & \begin{tabular}[x]{@{}c@{}}\% of vertices\\with $k(x) = 5$\end{tabular} & \begin{tabular}[x]{@{}c@{}}\% of vertices\\with $k(x) \geq 6$\end{tabular}\\ \hline

{\sc dutch-elite} & 6 & 2.35 & 14.9 & 0 & 54.3 & 0 & 29.1 & 0 & 1.7 \\
{\sc facebook} & 2 & 0.686 & 51.9 & 27.6 & 20.5 &&&&  \\
{\sc eva}  & 2 & 0.571 & 47.6 & 47.7 & 4.7 &&&&  \\
{\sc slashdot}    & 3 & 1.777 & 2.3 & 24.1 & 67.1 & 6.5  &&&  \\
{\sc loans}  & 3 & 2.06 & 0.1 & 13.9 & 66.3 & 19.7 && &\\
{\sc twitter} & 4 & 2.569 & 0.1 & $\approx 1$ & 44.4 & 51.2 & 3.4 && \\
\hline

{\sc email-virgili} & 4 & 2.729 & 0.1 & 2.3 & 32 & 55.7 & 9.9 &&  \\
{\sc email-enron} & 2 & 0.906 & 23.4 & 62.6 & 14 &&&&\\
{\sc email-eu} & 2 & 0.002 & 99.8 & 0.1 & 0.1 &&&&    \\
{\sc wikitalk-china} & 3 & 2.076 & $\approx 0$ & 0.01 & 92.4 & 7.6 &&& \\
\hline

{\sc ce-metabolic} & 3 & 1.982 & 0.2 & 7.5 & 86.1 & 6.2 &&& \\
{\sc sc-ppi} & 3 & 0.981 & 32.4 & 41.6 & 21.5 & 4.5 &&& \\
{\sc yeast-ppi} & 3 & 1.872 & 2 & 25.4 & 55.8 & 16.8 &&& \\
{\sc homo-pi}     & 2 & 0.747 & 34.2 & 56.9 & 8.9 &&&&  \\ \hline

{\sc as-graph-1} & 3 & 1.791 & 0.5 & 24.9 & 69.7 & 4.9 &&& \\
{\sc as-graph-2} & 3 & 1.124 & 9.6 & 68.5 & 21.7 & 0.2 &&&  \\
{\sc as-graph-3} & 2 & 0.828 & 27.8 & 61.6 & 10.6 &&&&  \\
{\sc routeview}   & 2 & 0.329 & 69.7 & 27.6 & 2.7 &&&&  \\
{\sc as-caida}    & 0 & 0 & 100 &&&&&&    \\
{\sc itdk}        & 4 & 2.108 & 0.3 & 12 & 64.5 & 22.8 & 0.4 &&   \\ \hline

{\sc gnutella-06} & 4 & 2.507 & 0.3 & 5.7 & 41.1 & 48.8 & 4.1 &&  \\
{\sc gnutella-24} & 5 & 2.697 & 0.2 & 1.5 & 37 & 50.7 & 10.5 & 0.1 &  \\
{\sc gnutella-30} & 5 & 3.167 & 0.1 & 1.8 & 13 & 52.4 & 31.8 & 0.9 & \\
{\sc gnutella-31} & 6 & 4.176 & 0.01 & 0.2 & 1.4 & 13.3 & 51.1 & 33.4 & 0.5 \\ \hline

{\sc web-stanford}&  28 & 0.006 & 99.9 & $\approx 0$ & $\approx 0$ & $\approx 0$ & $\approx 0$ & $\approx 0$ & $\approx 0$ \\
{\sc web-notredam}&  2 & 0.935 & 7.1 & 92.4 & 0.5 &&&& \\
{\sc web-berkstan}&  22 & 0.002 & 99.9 & $\approx 0$ & 0 & $\approx 0$ & 0 & $\approx 0$ & $\approx 0$   \\ \hline

{\sc amazon-1}    &  6 & 0.991 & 28.1 & 48 & 21.5 & 1.5 & 0.5 & 0.3 & 0.1 \\
{\sc amazon-2}    &  6 & 3.735 & 0.1 & 0.3 & 3.6 & 33.9 & 46.5 & 15.3 & 0.3 \\ \hline

{\sc road-euro}  &  8 & 0.135 & 97.4 & 0.3 & 0.1 & 0.4 & 0 & 0.8 & $\approx 1$ \\
{\sc openflight} &   3 & 1.879 & 0.2 & 23.9 & 63.7 & 12.2 &&& \\
{\sc power-grid} &   13 & 5.735 & 14.3 & 13.1 & 1.6 & 1.6 & 3.9 & 8.7 & 39.8  \\
{\sc road-pa} &   98 & 23.339 & 0.02 & 1.5 & 0.1 & 2.9 & 0.2 & 0.2 & 95 \\
\hline
\end{tabular}
\par } \medskip
\caption{Distribution  of values $k(x) = ecc_{T_1}(x) - ecc_G(x)$, $x \in V$.  $k_{max} := \max_{x\in V} k(x)$.  $k_{avg} := \frac{1}{n} \sum_{x \in V} k(x)$.}
\label{table:k-distributions-method1}
\end{center}
\vspace{0.5ex}
\commentout{
\raggedright Full distribution of distortions for networks with $k_{max > 6}$ ($|k_i|$: number of vertices with distortion $i$): \\
{\sc web-stanford} (no. of vertices with un-preserved eccentricities is 209): $|k_3| = 106$, $|k_5| = 59$, and $|k_{28}| = 17$.  \\

{\sc web-BerkStan} (no. of vertices with un-preserved eccentricities is 210): $|k_3| = 107$, $|k_5| = 60$, and $|k_{22}| = 23$.
{\sc power-grid} (no. of vertices with un-preserved eccentricities is 4233): $|k_1| = 649$, $|k_5| = 431$, $|k_{7}| = 516$, $|k_8| = 683$, $|k_9| = 703$, and $|k_{13}| = 46$.  \\

}
\vspace*{-1.0cm}
\end{table}

\begin{table}[htb]
\begin{center}{\fontsize{7}{9} \selectfont
\begin{tabular}{|l|cc|ccccccc|}
\hline
Network &  $k_{max}$ & $k_{avg}$ & \begin{tabular}[x]{@{}c@{}}\% of vertices\\with $k(x) = 0$\end{tabular} & \begin{tabular}[x]{@{}c@{}}\% of vertices\\with $k(x) = 1$\end{tabular} & \begin{tabular}[x]{@{}c@{}}\% of vertices\\with $k(x) = 2$\end{tabular} & \begin{tabular}[x]{@{}c@{}}\% of vertices\\with $k(x) = 3$\end{tabular} & \begin{tabular}[x]{@{}c@{}}\% of vertices\\with $k(x) = 4$\end{tabular} & \begin{tabular}[x]{@{}c@{}}\% of vertices\\with $k(x) = 5$\end{tabular} & \begin{tabular}[x]{@{}c@{}}\% of vertices\\with $k(x) \geq 6$\end{tabular}\\ \hline

{\sc dutch-elite} &   6 & 2.431 & 16.1 & 0 & 47.1 & 0 & 35.9 & 0 & 0.8  \\
{\sc facebook} &   3 & 0.704 & 43.6 & 42.5 & 13.8 & 0.1 &&&  \\
{\sc eva} &   2 & 0.572 & 47.6 & 47.6 & 4.8 &&&& \\
{\sc slashdot}    &   3 & 1.88 & 0.1 & 17.7 & 76.2 & $\approx 6$ &&&   \\
{\sc loans}  &   3 & 2.031 & 0.1 & 14 & 68.7 & 17.2 &&& \\
{\sc twitter} &   3 & 1.821 & 3.1 & $\approx 20$ & 68.6 & 8.3 &&& \\
\hline

{\sc email-virgili} &   4 & 1.932 & 4.3 & 22.8 & 48.4 & 24.4 & 0.1 &&  \\
{\sc email-enron} &   2 & 0.903 & 22.4 & 64.8 & 12.7 &&&& \\
{\sc email-eu} &   2 & 0.002 & 99.9 & 0.03 & 0.1 &&&&   \\
{\sc wikitalk-china} &   3 & 1.791 & $\approx 0$ & 21 & 79 & 0.008 &&& \\
\hline

{\sc ce-metabolic} &   1 & 0.349 & 65.1 & 34.9 &&&&& \\
{\sc sc-ppi} &   7 & 4.196 & 1.3 & 4.1 & 6.2 & 13.4 & 27.2 & 35.9 & 11.8 \\
{\sc yeast-ppi} &   4 & 2.558 & 0.7 & 5.9 & $\approx 36$ & 51.7 & 5.7 && \\
{\sc homo-pi}     &   2 & 0.612 & 41.6 & 55.5 & 2.9 &&&& \\ \hline

{\sc as-graph-1} &   2 & 0.887 & 19.6 & 72.2 & 8.2 &&&& \\
{\sc as-graph-2} &   2 & 0.833 & 25.7 & 65.3 & $\approx 9$ &&&& \\
{\sc as-graph-3} &   2 & 0.312 & 70.4 & 28 & 1.6 &&&&\\
{\sc routeview}    &   2 & 0.329 & 69.7 & 27.6 & 2.7 &&&& \\
{\sc as-caida}     &   0 & 0 & 100 &&&&&&   \\
{\sc itdk}  &   5 & 2.702 & 0.3 & 3.4 & 28.6 & 61.4 & 6.3 & $\approx 0$ &   \\ \hline

{\sc gnutella-06} &   5 & 3.543 & 0.01 & 0.7 & 5.9 & 37.2 & 50.9 & 5.3 &  \\
{\sc gnutella-24} &   6 & 4.475 & 0.02 & 0.1 & 0.7 & 8.6 & 38.3 & 46.5 & 5.7 \\
{\sc gnutella-30} &   6 & 4.034 & 0.02 & 0.2 & 2.6 & 16.4 & 54.8 & 25.1 & 0.5 \\
{\sc gnutella-31} &   6 & 4.251 & 0.01 & 0.1 & 1.3 & 11.6 & 48.4 & 37.8 & 0.9 \\ \hline

{\sc web-stanford}&   28 & 0.006 & 99.9 & $\approx 0$ & $\approx 0$ & 0.04 & $\approx 0$ & 0.02 & $\approx 0$  \\
{\sc web-notredam}&   2 &  0.935 & 7.1 & 92.3 & 0.6 &&&&\\
{\sc web-berkstan}&   22 & 0.002 & 99.97 & $\approx 0$ & 0 & 0.02 & 0 & 0.01 & $\approx 0$\\ \hline

{\sc amazon-1}    &   7 & 0.919 & 49.7 & 21.7 & 18.6 & 8.1 & 1.1 & 0.4 & 0.3 \\
{\sc amazon-2}    &   5 & 2.03 & 1.2 & 15.1 & 65 & 17.1 & 1.6 & $\approx 0$ & \\ \hline

{\sc road-euro}  &   8 & 0.135 & 97.4 & 0.3 & 0.1 & 0.4 & 0 & 0.8 & $\approx 1$ \\
{\sc openflight} &   2 & 0.641 & 36.1 & 63.7 & 0.2 &&&& \\
{\sc power-grid} &   4 & 1.409 & 46.3 & 13.1 & 12.6 & 9.1 & 18.8 &&   \\
{\sc road-pa} &   80 & 22.545 & 0.7 & 20.9 & 0.3 & 0.2 & 0.4 & 0.2 & 77.3 \\
\hline
\end{tabular} 
 \par } \medskip
\caption{Distribution  of values $k(x) = ecc_{T_2}(x) - ecc_G(x)$, $x \in V$.  $k_{max} := \max_{x\in V} k(x)$.  $k_{avg} := \frac{1}{n} \sum_{x \in V} k(x)$.}
\label{table:k-distributions-method2}
\end{center}
\commentout{
\vspace{0.5ex}
\raggedright Complete distribution: \\
{\sc web-stanford} $|k_3| = 106$, $|k_5| = 59$, and $|k_{28}| = 17$.  \\

{\sc web-BerkStan} $|k_3| = 107$, $|k_5| = 60$, and $|k_{22}| = 23$.
}
\end{table}

\end{document}